\documentclass  [journal,onecolumn,draft,tbtags,12pt]    {IEEEtran}

\usepackage{mathtools}
\usepackage{dsfont}
\allowdisplaybreaks
\usepackage{mathrsfs,amssymb,bm}

\usepackage[final]{graphicx}

\usepackage{booktabs}
\usepackage[compress]{cite}
\usepackage[final]{hyperref}
\hypersetup{
    colorlinks,%
    citecolor=black,%
    filecolor=black,%
    linkcolor=black,%
    urlcolor=black
}

\let\ALPHABET\mathcal
\let\POLICY  \mathbf

\usepackage{xstring}
\def\VEC#1%
    {\IfSubStr{ABCDEFGHIJKLMNOPQRSTUVWXYZabcdefghijklmnopqrstuvwxyz}{#1}
        {\mathbf{#1}}
        {\bm{#1}}}

\def\naturalnumbers{\mathds{N}}
\def\positivenumbers{\mathds{Z}_{+}}
\def\DEFINED{\coloneqq}
\def\BYDEFINITION{\eqqcolon}

\def\EXP{\mathop{\mathds{E}\kern0pt}\nolimits}
\def\IND{\mathop{\mathds{1}\kern0pt}\nolimits}
\def\D  {\mathord{\mathrm{d}\kern0pt}}
\makeatletter 
\def\PSP{\@ifstar\PSP@ii\PSP@i}
\def\PSP@i#1{\Delta(\ALPHABET #1)}
\def\PSP@ii#1{\Delta(#1)}
\makeatother
\def\PR {\mathop{\mathds{P}\kern0pt}\nolimits}
\def\gobble#1{}

\newtheorem{theorem}        {Theorem}
\newtheorem{proposition}    {Proposition}
\newtheorem{lemma}          [proposition] {Lemma}

\newtheorem{problem}        {Problem}

\newtheorem{definition}     {Definition}

\makeatletter
\def\env@cases{%
  \let\@ifnextchar\new@ifnextchar
  \left\lbrace
  \def\arraystretch{1}%
  \array{@{}l@{\quad}l@{}}}
\makeatother
\begin{document}

\title {Optimal decentralized control of coupled subsystems with control sharing}
\author{Aditya Mahajan%
        \thanks{Aditya Mahajan is with the Department of Electrical and Computer Engineering, 
                McGill University, Montreal, QC H3A 2A7, Canada. 
                \texttt{aditya.mahajan@mcgill.ca}}
        \thanks{Preliminary version of this paper appeared in the proceedings of
          the 48th Allerton conference on communication, control, and
          computation, 2010 and the 50th IEEE conference on decision and
          control, 2011.}
       \\
       \small\today}

\maketitle

\vspace*{-2\baselineskip}

\begin{abstract}
  Subsystems that are coupled due to dynamics and costs arise naturally in
  various communication applications. In many such applications the control
  actions are shared between different control stations giving rise to a
  \emph{control sharing} information structure. Previous studies of
  control-sharing have concentrated on the linear quadratic Gaussian setup and a
  solution approach tailored to continuous valued control actions. In this paper a
  three step solution approach for finite valued control actions is presented.
  In the first step, a person-by-person approach is used to identify redundant
  data or a sufficient statistic for local information at each control station.
  In the second step, the common-information based approach of Nayyar et al.\
  (2011) is used to find a sufficient statistic for the common information
  shared between all control stations and to obtain a dynamic programming
  decomposition. In the third step, the specifics of the model are used to
  simplify the sufficient statistic and the dynamic program. As an example, an
  exact solution of a two-user multiple access broadcast system is presented.
\end{abstract}


\section {Introduction}

\subsection {Motivation}

In this paper, we investigate a modular architecture for networked
control systems that consists of a collection of dynamically coupled subsystems, each with a
local control station. Each control station observes, either fully or
partially, the state
of its subsystem, but does not observe the state of other
subsystems.\footnote{In the formal description of the model
  Section~\ref{sec:model}, we also assume that, in addition to the local
  observations, the control stations also observe a shared state. With this
  slight generalization, the model can capture more applications, but does not
  add any additional conceptual difficulties. For that reason, we do not include
the shared state in the discussion in this section.} In addition, each control
station observes the control action of all other control stations with one-step
delay. Such a \emph{control sharing} happens naturally in applications
like multi-access broadcast~\cite{HluchyjGallager:1981, OoiWornell:1996} (see
Section~\ref{sec:MAB}), paging and registration in mobile cellular
systems~\cite{HajekMitzelYang:2008}, real-time communication with
feedback~\cite{WalrandVaraiya:1983}, and sensor networks~\cite{Veeravalli:1993}.

Each control station affects the state evolution of each substation; thus the
substation have coupled dynamics. The per unit cost depends on the state of all
subsystems and the control actions of all control stations; thus the control
stations are coupled through cost. No control station knows the
information available to other control stations. Hence, the system has a
non-classical information structure~\cite{Witsenhausen:1971, Ho:1980}.

Each control station has a perfect recall, that is, it chooses a control action
based on the history of its observations and control actions. Since the domain
of the control laws increases with time, we need to find a time-homogeneous
sufficient statistic for the past data at each controller to pose and solve the
infinite horizon optimal control problem. Finding such
sufficient statistics is difficult due to the non-classical nature of
information. For systems with classical information, a sufficient statistic at a
control station captures the affect of past data (at that control station) on
future estimation (at that control station). This feature is called the \emph{dual
affect} of control. For systems with non-classical information, in addition to the
above,  a sufficient statistic at a control station must capture the affect of
past control actions (at that control station) on the future estimation \emph{at
other control stations}. This feature is called the \emph{triple affect} of
control---the third affect being the \emph{signaling} affect. The control
sharing information structure makes the signaling affect explicit; as such
solution techniques to control sharing provide insights to other non-classical
information structures where the signaling affect is implicit. 

\subsection{Literature Overview}

There are only a few general frameworks of dynamic programming for systems with
non-classical information structure: the sequential team approach for finite
horizon systems~\cite{Witsenhausen:1973}, a common-information based approach
for finite horizon systems~\cite{Nayyar:phd}, and a two-step solution approach
for two-agent finite and infinite horizon systems~\cite{Mahajan:phd}. We are
interested in a solution framework that works for multiple control stations and
extends to infinite horizon systems, and hence, these generic dynamic programming
approaches are not applicable. 

Most of the research on non-classical information structure has focused on
specific system-dynamics and/or specific information structures. We briefly
describe some of these approaches below (see~\cite{YukselBasar:2013} for a
detailed discussion.)

The special case of linear dynamics and Gaussian disturbances with specific
information structures has received considerable attention in the literature.
Examples include static teams~\cite{Radner:1962, KrainakSpeyerMarcus:1983},
partially nested teams~\cite{HoChu:1972, Gattami:phd}, stochastically nested
teams~\cite{Yuksel:2009}, and quadratic invariant
teams~\cite{BamiehVoulgaris:2005,RotkowitzLall:2006, Rantzer:2006}
We are interested in systems with non-linear dynamics. More
importantly, the control-sharing information structure is neither static, 
nor partially nested, nor stochastically nested, nor quadratically invariant.

The special case of non-classical information structure with specific data sharing 
between the control stations has also received considerable attention in the
literature. Examples include delayed-state
observation~\cite{AicardiDavoliMinciardi:1987}, delayed (observation)
sharing~\cite{Witsenhausen:1971, VaraiyaWalrand:1978,
NayyarMahajanTeneketzis:2011}, control sharing~\cite{Bismut:1972,
SandellAthans:1974}, periodic sharing~\cite{OoiVerboutLudwigWornell:1997},
belief sharing~\cite{Yuksel:2009}, and partial history
sharing~\cite{NayyarMahajanTeneketzis:2012}.
Out of these, the models closest to our setup are control-sharing and partial
history sharing.

As described earlier, in a control-sharing information structure, control
stations can directly signal to one another through their control actions. This
signaling aspect was exploiting in~\cite{Bismut:1972, SandellAthans:1974} by
explicitly embedding the local observations in the control actions with
arbitrary small perturbation of the control action. Their embedding technique relies on:
(i)~real-valued random variables have infinite information (in an information
theoretic sense); and (ii)~the existence of measurable bijections between
Euclidean spaces. Such an embedding of observations converts the control sharing
information structure to a one-step delayed (observation) sharing information
structure, which is also a partially nested information structure. Then, the
solution techniques for partially nested teams give an approximate solution for
the control-sharing information structure~\cite{SandellAthans:1974}. However,
our motivation for investigating these models comes from communication
networks, most of which have finite valued control actions.\footnote{Even
  otherwise, the assumption of noiseless sharing continuous valued control
  actions is not realistic in communication applications because it requires
  infinite capacity communication channels.} Embedding observations in finite
valued control actions is not possible. Hence, the solution technique
of~\cite{Bismut:1972, SandellAthans:1974} does not work for finite valued action
spaces.

In a system with control sharing, each control station knows part of the history
of data at all control stations. Thus, a control-sharing information structure
is also a \emph{partial history sharing} information structure, for which the
following solution approach is proposed in~\cite{NayyarMahajanTeneketzis:2012}.
Split the data available at each control station into two parts: a common
information part that is commonly shared amongst all control stations, and a
local information part that consists of the remaining data. Then, the
decentralized stochastic control problem is equivalent to a centralized
centralized stochastic control problem in which a fictitious \emph{coordinator}
observes the common information and chooses functions that map the local
information at each control station to its action. This solution approach
extends to infinite horizons only when the location information is not
increasing with time---which is not the case in the above model because the local
information, which is the history of local state observations, is increasing
with time.
Hence, the solution approach of~\cite{NayyarMahajanTeneketzis:2012} is not
directly applicable to control sharing information structures.

\subsection {Contributions of the paper}

One of the main difficulties in obtaining dynamic programming decomposition for
decentralized stochastic control is to identify sufficient statistics (for each
control station) that summarize the affect of the history of their observations
and actions on future observations and cost. In this paper, we present a three
step approach to find such sufficient statistics for decentralized control of
dynamically coupled subsystems with control sharing. 

In the first step, we use a person-by-person approach and identify either
irrelevant data or a sufficient statistic for part of the data at each control
station. In the second step, we use the common information approach
of~\cite{NayyarMahajanTeneketzis:2012} and identify sufficient statistic for the
common information at all control stations. In the third step, we use the
salient features of the model---full or partial observation of local states,
dynamic coupling using control actions, and sharing of control actions---to
simplify the sufficient statistic obtained in the second step. Using the
sufficient statistics of the second step (and their simplification in the third
step), we obtain a dynamic programming decomposition which can be extended to
infinite horizon discounted cost setup. Such a dynamic programming decomposition
is not possible by using either the person-by-person approach or the
common-information approach alone.

We use the proposed solution approach to obtain a dynamic programming
decomposition for a multiuser broadcast channel, and analytically solve the
dynamic program when both users have the same arrival rates. Although this
example is very well studied, this is the first result that provides a dynamic
programming decomposition for this model.

The rest of this paper is organized as follows. We present two models for
coupled subsystems with control sharing in Section~\ref{sec:coupled}: 
the full and partial observation models. We present the above described three
step approach for the full observation model in Section~\ref{sec:full} and for
the partial observation model in Section~\ref{sec:partial}. We present an
example of a two-user multiaccess broadcast channel in Section~\ref{sec:MAB} and
conclude in Section~\ref{sec:conclusion}.

\subsection {Notation}

Random variables are denoted with upper case letters ($X$, $Y$, etc.), their
realization with lower case letters ($x$, $y$, etc.), and their space of
realizations by script letters ($\ALPHABET X$, $\ALPHABET Y$, etc.). Subscripts
denote time and superscripts denote the subsystem; e.g. $X^i_t$ denotes the
state of subsystem~$i$ at time~$t$.
The short hand notation $X^i_{1:t}$ denotes the vector $(X^i_1, X^i_2, \dots,
X^i_t)$. Bold face letters denotes the collection of variables at all
subsystems; e.g., $\VEC X_t$ denotes $(X^1_t, X^2_t, \dots X^n_t)$. The
notation $\VEC X^{-i}_t$ denotes the vector $(X^1_t, \dots, X^{i-1}_t,
X^{i+1}_t, \dots, X^n_t)$. 

$\PSP X$ denotes the probability simplex on the space $\ALPHABET X$.
$\PR(A)$ denotes the probability of an event $A$, and $\EXP[X]$ denotes the
expectation of a random variable $X$. $\IND[x=y]$ denotes the indicator function
of the statement $x=y$, i.e., $\IND[x=y] = 1$ if $x=y$ and $0$ otherwise. 
Let $\naturalnumbers$ denote the set of
natural numbers and $\positivenumbers$ denote the set of non-negative integers.  

\section {Coupled subsystems with control sharing} \label{sec:coupled}
\subsection {Model and Problem Formulation}
\label{sec:model}


\subsubsection*{System components}

Consider a discrete-time networked control system with $n$ subsystems.
The state
$(Z_t, X^i_t)$ of subsystem $i$, $i=1,\dots,n$, has two components: a
\emph{local state} $X^i_t \in \ALPHABET X^i$ and a \emph{shared state} $Z_t \in
\ALPHABET Z$, which is identical for all subsystems. The initial shared state
$Z_1$ has a distribution $P_Z$. Conditioned on the initial shared state $Z_1$,
the initial local state of all subsystems are independent; initial local state
$X^i_1$ is distributed according to $P_{X^i|Z}$, $i=1,\dots,n$. Let $\VEC X_t
\DEFINED (X^1_t, \dots, X^n_t)$ denote the local state of all subsystems.

A control station is co-located with each subsystem. Let $U^i_t \in \ALPHABET
U^i$ denote the control action of control station~$i$ and $\VEC U_t \DEFINED
(U^1_t, U^2_t, \dots, U^n_t)$ denote the collection of all control actions. 

\subsubsection* {System dynamics}
The shared and the local state of each subsystems 
are coupled through the control actions; the shared state evolves according to
\begin{equation} \label{eq:shared}
  Z_{t+1} = f^0_t(Z_t, \VEC U_t, W^0_t)
\end{equation}
while the local state of subsystem~$i$,
$i=1,\dots,n$, evolves according to:
\begin{equation}\label{eq:dynamics}
  X^i_{t+1} = f^i_t(Z_t, X^i_t, \VEC U_t, W^i_t)
\end{equation}
where $W^i_t \in \ALPHABET
W^i$, $i=0,1,\dots,n$, is the plant disturbance with distribution $P_{W^i}$. The processes
$\{W^i_t, t=1,\dots\}$, $i=0,1,\dots,n$, are assumed to be independent across
time, independent of each other, and also independent of the initial state $(Z_1,
\VEC X_1)$ of the system.

Note that the updated local state of subsystem~$i$ depends only on the previous
local state of subsystem~$i$ and previous shared state but is controlled by all
control stations. 

\subsubsection*{Observation models and information structures}

We consider two observation models that differ in the observation of the
location state $X^i_t$ at control station~$i$. In the first model, called
\emph{full observation model}, control station~$i$ perfectly observes the local
state $X^i_t$; in the second model, called \emph{partial observation model},
control station~$i$ observes a noisy version $Y^i_t \in \ALPHABET Y^i$ of the
local state $X^i_t$ given by
\begin{equation}
  Y^i_t = \ell^i_t(X^i_t, \tilde W^i_t)
\end{equation}
where $\tilde W^i_t \in \tilde {\ALPHABET W}^i_t$ is the observation noise
with distribution $P_{\tilde W^i}$. 
The processes $\{\tilde W^i_t$, $t=1,\dots,\}$, $i=1,\dots,n$ are assumed to
be independent across time, independent of each other, independent of
$\{W^i_t$, $t=1,\dots,n\}$, and independent of the initial states $(\VEC
X_1, Z_1)$. 

In both models, in addition to the local measurements of the state of its
subsystem, each control station perfectly observes the shared state $Z_t$
and the one-step delayed control actions $\VEC U_{t-1}$ of all control stations.
The control stations perfectly recall all the data they observe. Thus, in
the full observation model, control station~$i$ chooses a control action
according to
\begin{equation}\label{eq:M1}
  U^i_t = g^i_t(Z_{1:t}, X^i_{1:t}, \VEC U_{1:t-1})
\end{equation}
while in the partial observation model, it chooses a control action according to
\begin{equation}\label{eq:M2}
  U^i_t = g^i_t(Z_{1:t}, Y^i_{1:t}, \VEC U_{1:t-1}).
\end{equation}

The function $g^i_t$ is called the \emph{control law} of control station~$i$.
The collection $\POLICY g^i \DEFINED (g^i_1, g^i_2, \dots, g^i_T)$ of control
laws at control station~$i$ is called the \emph{control strategy of control
station~$i$}. The collection $\POLICY g \DEFINED (\VEC g^1, \VEC g^2, \dots,
\VEC g^n)$ of control strategies of all control stations is called the
\emph{control strategy of the system}. 

\subsubsection*{Cost and performance}

At time $t$, the system incurs a cost $c_t(Z_t, \VEC X_t, \VEC U_t)$ that
depends on the shared state, the local state of all subsystems, and the actions
of all control stations. Thus, the subsystems are also coupled through cost. 

The system
runs for a time horizon $T$. 
The performance of a control strategy $\POLICY g$ is measured by the
expected total cost incurred by that strategy, which is given by
\begin{equation}\label{eq:cost}
  J(\POLICY g) \DEFINED \EXP\Big[ \sum_{t=1}^T c_t(Z_t, \VEC X_t, \VEC U_t) \Big]
\end{equation}
where the expectation is with respect to a joint measure of $(Z_{1:T}, \VEC X_{1:T}, \VEC
U_{1:T})$ induced by the choice of the control strategy $\VEC g$.

We are interested in the following optimal control problem:
\begin{problem} \label{prob:finite}
  Given the distributions $P_Z$, $P_{X^i|Z}$, $P_{W^i}$, $P_{\tilde W^i}$ of the initial shared
  state, initial local state, plant disturbance of subsystem~$i$, and
  observation noise of subsystem~$i$ (for the partial observation model), $i=1,\dots,n$, a horizon
  $T$, and the cost functions $c_t$, $t=1,\dots, T$, find a control strategy
  $\VEC g$ that minimizes the expected total cost given by~\eqref{eq:cost}.
\end{problem}

\subsection {Applications in communication networks} \label{sec:applications}

Control-sharing information structure arises naturally in communication
networks, as is illustrated by some applications described below. 

\subsubsection {Paging and registration in cellular networks}

Consider a mobile cellular network consisting of two controllers: a network
operator and a mobile station. The local state $X^1_t$ of the network operator
is a constant and the local state $X^2_t$ of the mobile station is its current
location that changes in a Markovian manner. We will describe the shared state
$Z_t$ later. The control action $U^1_t$ of the network operator is a permutation
of $\ALPHABET X^2$, the set of all possible locations of the mobile, and denotes the order in which mobile station will be
searched if there is a paging request. The control action $U^2_t$ of the mobile
station is either $X^2_t$ (indicating that the mobile station registers with the
network) or $\mathsf{NR}$ (indicating the mobile station does not register). 

At each time, the network may get an exogenous \emph{paging} request to seek the
location of the mobile station. If a paging request is received (denoted by $P_t
= 1$), the cost of searching is given by the index $i(X^2_t, U^1_t)$ of $X^2_t$
in $U^1_t$. If no paging request is received (denoted by $P_t = 0$) and the
mobile station registers with the network, a registration cost of $r$ is
incurred. The process $P_t$ is a binary-valued Markov process. If either the
mobile station is paged or the mobile station registers with the network, the
network operator learns the current location of the mobile station. Let $M_t$
denote the time since the last paging or registration request and $S_t$ denote
the location of the mobile station at that time. Then $Z_t = (P_t, M_{t-1},
S_{t-1})$ is the shared state of the system.

The above model corresponds to the model of paging and registration in mobile
cellular network considered in~\cite{HajekMitzelYang:2008}. The control action
$U^1_t$ of the network is based on information known to the mobile station,
hence $U^1_t$ is effectively observed at the mobile station. The control action
$U^2_t$ of the mobile station is communicated to the network operator. Hence,
this system has control-sharing information structure. 

\subsubsection {Real-time communication}

Consider a real-time communication system consisting of an encoder and a
decoder. The encoder observes a first-order Markov source $S_t$. 
The local state $X^1_t$ of the encoder is $(S_{t-1}, S_t)$ and the local state
$X^2_t$ of the decoder is a constant. The shared state $Z_t$ is also a constant.
The control action $U^1_t$ of the encoder is a quantization symbol that is
communicated to the decoder. The control action $U^2_t$ of the decoder is an
estimate of the one-step delayed source $S_{t-1}$ of the encoder. The cost at
each time is given by a distortion between $S_{t-1}$ and $U^2_t$. 

The above model corresponds to the model of real-time communication considered
in~\cite{Witsenhausen:1979} (specialized to infinite memory). The control action
$U^1_t$ of the encoder is communicated to the decoder. The control action
$U^2_t$ of the decoder is based on the information known to the encoder, hence
$U^2_t$ is effectively observed at the encoder. Hence, this system has
the full observation model considered above with shared state $Z_t = \emptyset$.

\subsubsection {Multiaccess broadcast}

Consider a two-user multiaccess broadcast system. At time $t$, $W^i_t \in
\{0,1\}$ packets arrive at each user according to independent Bernoulli
processes with $\PR(W^i_t = 1) = p^i$, $i=1,2$. Each user may store only $X^i_t
\in \{0,1\}$ packets in a buffer. If a packet arrives when the user-buffer is full,
the packet is dropped. 

Both users may transmit $U^i_t \in \{0,1\}$ packets over a shared broadcast
medium. A user can transmit only if it has a packet, thus $U^i_t \le X^i_t$. If
only one user transmits at a time, the transmission is successful and the
transmitted packet is removed from the queue. If both users transmit
simultaneously, packets ``collide'' and remain in the queue. Thus, the state
update for user~1 is given by
\( X^1_{t+1} = \max( X^1_t + U^1_t \cdot (1- U^2_t) + W^1_t, 1). \)
The state update rule for user~2 is symmetric dual of the above.  

Instead of costs, it is more natural to work with rewards in this example. The
objective is to maximize throughput, or the number of successful packet
transmissions. Thus, the per unit reward is $c(\VEC x, \VEC u) = u^1 \oplus
u^2$, where $\oplus$ means binary XOR. 

When the arrival rates at both users are the same ($p^1 = p^2$), the above model
corresponds to the two-user multiaccess broadcast system considered
in~\cite{HluchyjGallager:1981, OoiWornell:1996, MahajanNayyarTeneketzis:2008}.
Slight variation of the above model were considered in~\cite{Schoute:1978,
VaraiyaWalrand:1979}. In recent
years, the two-user multiaccess broadcast system with asymmetric arrivals ($p^1
\neq p^2$) has been used as a benchmark problem for decentralized stochastic
control problems in the artificial intelligence
community~\cite{HansenBernsteinZilberstein:2004, SzerCharpilletZilberstein:2005,
SzerCharpillet:2005, SeukenZilberstein:2007, DibangoyeMouaddibCahibdraa:2008}.

Due to the broadcast nature of the communication channel, each user observes the
transmission decision of the other user. Hence the system has the full observation model
considered above with shared state $Z_t = \emptyset$. We will revisit this
model in Section~\ref{sec:MAB}.

\section {Main result for the full observation model} \label{sec:full}

In this section, we derive structure of optimal control laws and a dynamic
programming decomposition for the full observation model. As stated in the introduction, the
full observation model has a partial history sharing information
structure~\cite{NayyarMahajanTeneketzis:2012}. Nayyar \emph{et al.}\ proposed a
\emph{common information based} approach to design systems with partial
information sharing. According to their approach, the design of optimal control
strategies is investigated from the point of view of a coordinator that observes
the shared common information. In the full observation model, the shared common
information $C_t = (Z_{1:t}, \VEC U_{1:t-1})$, and the private local information
is $L^i_t = \{X^i_{1:t}\}$. According to \cite{NayyarMahajanTeneketzis:2012}, the
posterior probability $\PR(Z_t, \VEC L_t \mid C_t)$ is a sufficient
statistic for the shared common information $C_t$.

However, directly using the above approach is not useful for the full
observation model because the local information $L^i_t$ at
control station~$i$, $i=1,\dots,n$, is increasing with time, which causes the dimension of
the sufficient statistic $\PR(Z_t, \VEC L_t \mid C_t)$ to increases with time;
and therefore, $\PR(Z_t, \VEC L_t \mid C_t)$ does not work as a sufficient
statistic for infinite horizon setup.

In this paper, we present the following three step approach to simplify the
structure of the control laws and derive a dynamic programming decomposition
(that extends to the infinite horizon setup).

\begin{enumerate}
  \item Use a person-by-person approach to show that the past values of the
    local state $X^i_{1:t-1}$ are irrelevant at control station $i$ at time $t$.
    Thus, for any control strategy of control station~$i$ that uses $(X^i_{1:t},
    Z_{1:t}, \VEC U_{1:t-1})$, we can choose a control strategy that uses only
    $(X^i_t, Z_{1:t}, \VEC U_{1:t-1})$ without any loss in performance.

  \item When attention is restricted to control strategies of the form derived
    in Step~1, the local information $L^i_t = \{X^i_t\}$ at control station~$i$
    does not increase with time. Thus, using the results
    of~\cite{NayyarMahajanTeneketzis:2012}, we can show that $\Pi_t =
    \PR(\VEC X_t, Z_t \mid C_t)$ is a sufficient statistic for the common
    information $C_t$ and is also an information state for dynamic programming. 

  \item Using the system dynamics, show that $\Pi_t$ defined in Step~2 is
    equivalent to $(Z_t, \VEC \Theta_t)$, where $\VEC \Theta_t = (\Theta^1_t, \dots,
    \Theta^n_t)$ and $\Theta^i_t = \PR(X^i_t \mid C_t)$. Using this equivalence, we
    can simplify the structural result and dynamic programming decomposition of
    Step~2.
\end{enumerate}

Now, we describe each of these steps in detail. For simplicity of exposition, we
assume that $\ALPHABET Z$, $\ALPHABET X^i$, $\ALPHABET U^i$, and $\ALPHABET
W^i$, $i=1,\dots,n$, are finite. The results extend to general alphabets under
suitable technical conditions (similar to those for centralized stochastic
control~\cite{HernandezLermaLasserre:1996}).

\subsection*{Step 1: Shedding of irrelevant information}

In this section, we show that the past values of local state
$X^i_{1:t-1}$ are irrelevant at control station~$i$ at time~$t$, $i=1,\dots,n$.
In particular:
\begin{proposition}\label{prop:structure}
  In the full observation model, restricting attention to control laws of the form
  \begin{equation} \label{eq:controller-irrelevant}
    U^i_t = \bar g^i_t(X^i_t, Z_{1:t}, \VEC U_{1:t-1})
  \end{equation}
  at all control stations~$i$, $i=1,\dots,n$, is without loss of optimality.
\end{proposition}

A priori, it is not obvious that the past data is irrelevant. Suppose we pick any control
station~$i$, $i=1,\dots, n$; arbitrarily fix the control strategy of all control
stations except station~$i$ and consider the subproblem of finding the optimal
control strategy at control station~$i$. In principle, the history $X^i_{1:t-1}$
of local states at control station~$i$ may give some information about the
history $X^j_{1:t-1}$ of local states at control station~$j$, $j\neq i$; and
hence, may help in predicting the future control actions of control station~$j$.
The following proposition shows that this is not the case. Conditioned on the shared
observations $(Z_{1:t}, \VEC U_{1:t})$, the local state processes $\{X^i_t$,
$t=1,\dots\}$, $i=1,\dots,n$, evolve independently. 

\begin{proposition} \label{prop:independent}
  In the full observation model, the local states of all subsystems are
  conditionally independent given the history of shared state and control
  actions. Specifically, for any
  realization $z_t \in \ALPHABET Z$, $x^i_t \in \ALPHABET X^i$ and $u^i_t \in \ALPHABET U^i$ of $X^i_t$
  and $U^i_t$, $i=1,\dots,n$, $t=1,\dots,T$, we have
  \begin{equation}
    \PR(\VEC X_{1:t} = \VEC x_{1:t} \,|\, Z_{1:t} = z_{1:t}, \VEC U_{1:t} = \VEC u_{1:t}) 
    = \prod_{i=1}^n 
    \PR(X^i_{1:t} = x^i_{1:t} \,|\, Z_{1:t} = z_{1:t}, \VEC U_{1:t} = \VEC u_{1:t})
  \end{equation}
\end{proposition}
See Appendix~\ref{app:independent} for proof.
One immediate consequence of the above Proposition is the following:

\begin{lemma}\label{lemma:CMP}
  Consider the full observation model for an arbitrary but fixed
  choice of control strategy $\VEC g$. Define $R^i_t = (X^i_t, Z_{1:t}, \VEC U_{1:t-1})$.
  Then,
  \begin{enumerate}
    \item The process $\{R^i_t, t=1,\dots,T\}$ is a controlled Markov process with
      control action $U^i_t$, i.e., for any $x^i_t, \tilde x^i_t \in \ALPHABET
      X^i$, $z_t, \tilde z_t \in \ALPHABET Z$, $u^i_t, \tilde u^i_t \in
      \ALPHABET U^i$, $r^i_t = (x^i_t, z_{1:t}, \VEC
      u_{1:t-1})$, $\tilde r^i_t = (\tilde x^i_t, \tilde z_{1:t}, \tilde {\VEC u}_{1:t-1})$,
      $i=1,\dots,n$, and $t=1,\dots,T$, 
      \begin{equation*}
        \PR(R^i_{t+1} = \tilde r^i_{t+1} \,|\, R^i_{1:t} = r^i_{1:t}, U^i_{1:t} = u^i_{1:t}) 
        =
        \PR(R^i_{t+1} = \tilde r^i_{t+1} \,|\, R^i_t = r^i_t, U^i_{t} = u^i_{t})
      \end{equation*}

    \item The instantaneous conditional cost simplifies as follows:
      \begin{equation*}
        \EXP[ c_t(Z_t, \VEC X_t, \VEC U_t) \,|\,  R^i_{1:t} = r^i_{1:t}, U^i_{1:t} = u^i_{1:t}] 
        = \EXP[ c_t(Z_t, \VEC X_t, \VEC U_t) \,|\,  R^i_{t} = r^i_{t}, U^i_{t} =
        u^i_{t}]
      \end{equation*}
  \end{enumerate}
\end{lemma}
See Appendix~\ref{app:CMP} for proof.

In light of Lemma~\ref{lemma:CMP}, lets reconsider the subproblem of finding the
optimal control strategy for control station~$i$ when the control strategy $\VEC
g^{-i}$ of all other control stations is fixed arbitrarily. In this subproblem,
control station~$i$ has access to $R^i_{1:t}$, chooses $U^i_t$, and incurs an
expected instantaneous cost $\EXP[ c_t(\VEC X_t, \VEC U_t) \,|\, R^i_{1:t},
U^i_{1:t}]$. Lemma~\ref{lemma:CMP} implies that the optimal choice of control
strategy $\VEC g^i$ is a Markov decision process. Thus, using Markov decision
theory~\cite{Whittle:1983}, we get the following (recall that $R^i_t = (X^i_t,
Z_{1:t}, \VEC U_{1:t-1})$ and the choice of $\VEC g^{-i}$ is arbitrary):
\begin{lemma} \label{lemma:structure}
  Consider the full observation model for any arbitrary but fixed choice of control strategy
  $\VEC g^{-i}$ of all control stations except $i$. Then, restricting attention to
  control laws of the form
  \begin{equation} \label{eq:structure-local}
    U^i_t = g^i_t(X^i_t, Z_{1:t}, \VEC U_{1:t-1})
  \end{equation}
  at control station~$i$ is without loss of optimality.
\end{lemma}

\begin{proof}[Proof of Proposition~\ref{prop:structure}]
  Lemma~\ref{lemma:structure} implies that for an arbitrary choice of $\VEC
  g^{-i}$, control strategies of the form~\eqref{eq:structure-local} at control
  station~$i$ dominate those of the form~\eqref{eq:M1}. Cyclically using the
  same argument for all control stations proves the result.
\end{proof}

Even after shedding $X^i_{1:t-1}$, the data at each control station is still
increasing with time. In the next step, we show how to ``compress'' this data
into a sufficient statistic. 

\subsection* {Step 2: Sufficient statistic for common data}

Consider Problem~\ref{prob:finite} for the full observation model and restrict  control
strategies of the form~\eqref{eq:controller-irrelevant}.
Proposition~\ref{prop:structure} shows that this restriction is without loss of
optimality. We use the results of~\cite{NayyarMahajanTeneketzis:2012} for this
restricted setup.

Split the data at each control station into two parts: the common data $C_t =
(Z_{1:t}, \VEC U_{1:t-1}$) that is observed by all control stations and the
local (or private) data $L^i_t = X^i_t$ that is observed by only control
station~$i$. Note that the common information $C_t \subset C_{t+1}$ is
increasing with time, while the local information $L^i_t$ has a fixed size.
Thus, the system has \emph{partial history sharing} information structure with
finite local memory. Nayyar \emph{et al.}~\cite{NayyarMahajanTeneketzis:2012} derived
structural properties of optimal controllers and a dynamic programming
decomposition for such an information structure.

To present the result, we first define the following:
\begin{definition} \label{def:pi}
  Given any control strategy $\bar {\VEC g}$ of the
  form~\eqref{eq:controller-irrelevant}, 
  let $\Pi_t$, $t=1,\dots,T$, denote the
  posterior probability of $(Z_t, \VEC X_t)$ given the common information $C_t$;
   i.e., for any $z \in \ALPHABET Z$ and $x^i \in \ALPHABET X^i$, the component
   $(z, \VEC x)$ of $\Pi_t$ is given by 
   \[ \Pi_t(z, \VEC x) \DEFINED \PR^{\bar {\VEC g}}(Z_t = z, \VEC X_t = \VEC x \mid C_t). \]
\end{definition}
The update of $\Pi_{t}$ follows the standard non-linear filtering equation.
It is shown in~\cite{NayyarMahajanTeneketzis:2012} that $\Pi_t$ is a sufficient
statistic for $C_t$; in particular, we have the following structural result. 

\begin{proposition} 
  [{\cite[Theorem 2]{NayyarMahajanTeneketzis:2012}} applied to model of
  Proposition~\ref{prop:structure}]
  \label{prop:coordinator-structure}
  In the full observation model, restricting attention to control laws of the form
  \begin{equation} \label{eq:controller-common}
    U^i_t = \hat g^i_t(X^i_t, \Pi_t)
  \end{equation}
  at all control stations~$i$, $i=1,\dots,n$, is without loss of optimality.
\end{proposition}

To obtain a dynamic programming decomposition to find optimal control strategies of the
form~\eqref{eq:controller-common}, the following 
\emph{partially evaluated control laws} were defined
in~\cite{NayyarMahajanTeneketzis:2012}: For any control strategy of
the form~\eqref{eq:controller-common}, and any realization $\pi_t$ of $\Pi_t$,
let
\[ \hat d^i_t(\cdot) = \hat g^i_t(\cdot, \pi_t) \]
denote a mapping from $\ALPHABET X^i_t$ to $\ALPHABET U^i_t$. When $\Pi_t$ is a
random variable, the above mapping is a random mapping denoted by $\hat D^i_t$.
Let $\hat {\VEC d}_t = (\hat d^1_t, \dots, \hat d^n_t)$ and $\hat {\VEC D}_t =
(\hat D^1_t, \dots, \hat D^n_t)$. Then optimal control strategies of the
form~\eqref{eq:controller-common} are obtained as follows.

\begin{proposition}
  [{\cite[Theorem 3]{NayyarMahajanTeneketzis:2012}} applied to the model of
  Proposition~\ref{prop:structure}]
 \label{prop:coordinator-DP}
  For any $\pi_t \in \PSP*{\ALPHABET Z \times \ALPHABET X^1 \times \cdots
  \ALPHABET X^n}$, define 
  \begin{align}
    V_T(\pi_T) &= \min_{\hat{\VEC d}_T}
    \EXP [ c_T(Z_T, \VEC X_T, \VEC U_T)
      \mid \Pi_T = \pi_T, \hat{\VEC D}_T = \hat{\VEC d}_T ] \\
   \intertext{and for $t=T-1,T-2,\dots,1$,}
   V_t(\pi_t) &= \min_{\hat{\VEC d}_t}
   \EXP [ c_t(Z_t, \VEC X_t, \VEC U_t) + V_{t+1}(\Pi_{t+1})
      \mid \Pi_t = \pi_t, \hat{\VEC D}_t = \hat{\VEC d}_t ]
  \end{align}
  Let $\hat \Psi_t(\pi_t)$ denote the $\arg \min$ of the right hand side of
  $V_t(\pi_t)$, and $\hat \Psi^i_t$ denote the $i$-th component of $\hat \Psi_t$. 
  Then, a control strategy
  \[ \hat g^i_t(x^i_t, \pi_t) \in \hat \Psi^i_t(\pi_t)(x^i_t) \]
  is optimal for Problem~\ref{prob:finite} with the full observation model.
\end{proposition}
  
\subsection* {Step 3: Simplification of the sufficient statistic}

In this step, we use Proposition~\ref{prop:independent} to simplify the sufficient
statistic $\Pi_t$ used in Step~2, and thereby simplify
Propositions~\ref{prop:coordinator-structure} and~\ref{prop:coordinator-DP}. For
that matter, we define the following.
\begin{definition} \label{def:theta}
  Given any control strategy $\hat {\VEC g}$ of the form~\eqref{eq:controller-common}, 
  let $\Theta^i_t$, $t=1,\dots,T$, denote the
  posterior probability of $X^i_t$ given the common information $C_t$,
   i.e., for any $x^i \in \ALPHABET X^i$, the component
   $x^i$ of $\Theta^i_t$ is given by 
   \[ \Theta^i_t(x^i) \DEFINED \PR^{\hat {\VEC g}}( X^i_t = x^i \mid C_t). \]
\end{definition}
The update of $\Theta^i_t$ follows the standard non-linear filtering equation.
For completeness, we describe this update below.
\begin{lemma} \label{lemma:theta}
  There exists a deterministic function $F_t$ such that
  \begin{equation} \label{eq:theta-update}
    \VEC \Theta_{t+1} = F_t(\VEC \Theta_t, Z_{t+1}, \VEC U_t, \hat {\VEC D}_t)
  \end{equation}
\end{lemma}
The proof follows from the law of total probability and Bayes rule. See
Appendix~\ref{app:theta-update}.

We can now simplify the sufficient statistic $\Theta_t$ as follows:
\begin{lemma} \label{lemma:equivalent}
  For any $z \in \ALPHABET Z$, $x^i \in \ALPHABET X^i$, $i=1,\dots, n$, the
  values $(z,
  \VEC \Theta_t(\VEC x))$ are sufficient to compute $\Pi_t(z, \VEC x)$. 
\end{lemma}
\begin{proof}
  The proof follows directly from the definition of  $\Pi_t$, $\Theta^i_t$ and
  Proposition~\ref{prop:independent}. Let $C_t = (Z_{1:t}, \VEC U_{1:t-1}$ and
  consider the component $(z, \VEC x)$ of $\Pi_t$:
  \begin{align*}
    \Pi_t(z, \VEC x) 
    \stackrel{(a)}= \IND[Z_t = z] \cdot
    \PR(\VEC X_t = \VEC x  \mid Z_{1:t}, \VEC U_{1:t-1}) 
    \stackrel{(b)}= \IND[Z_t = z] \cdot
    \prod_{i=1}^n \Theta^i_t(x^i)
  \end{align*}
  where $(a)$ follows form the law of total probability and $(b)$ follows from
  Proposition~\ref{prop:independent}. 
\end{proof}

By substituting $(Z_t, \VEC \Theta_t)$ instead of $\Pi_t$ in
Propositions~\ref{prop:coordinator-structure} and~\ref{prop:coordinator-DP}, we
get the following:

\begin{theorem}[Structure of optimal controllers]
  \label{thm:structure}
  In the full observation model, restricting attention to control laws of the form
  \begin{equation} \label{eq:controller-final}
    U^i_t = \tilde g^i_t(X^i_t, Z_t, \VEC \Theta_t)
  \end{equation}
  at all control stations~$i$, $i=1,\dots,n$, is without loss of optimality.
\end{theorem}

For any control strategy of the form~\eqref{eq:controller-final}, and any
realization $\VEC \theta_t$ of $\VEC \Theta_t$, 
let
\[ \tilde d^i_t(\cdot) = \tilde g^i_t(\cdot, z_t, \VEC {\theta}_t) \]
denote a mapping from $\ALPHABET X^i_t$ to $\ALPHABET U^i_t$. When $\VEC \Theta_t$ is a
random variable, the above mapping is a random mapping denoted by $\tilde D^i_t$. Let
$\tilde {\VEC d}_t = (\tilde d^1_t, \dots, \tilde d^n_t)$ and $\tilde {\VEC D}_t
= (\tilde D^1_t, \dots, \tilde D^n_t)$. Then optimal control strategies of the
form~\eqref{eq:controller-common} are obtained as follows.

\begin{theorem}[Dynamic programming decomposition]
  \label{thm:DP}
  For any $z_t \in \ALPHABET Z$ and $\theta^i_t \in \PSP*{\ALPHABET X^i}$,
  $i=1,\dots,n$, define
  \begin{align}
    V_T(z_T, \VEC \theta_T) &= \min_{\tilde{\VEC d}_T}
    \EXP [ c_T(Z_T, \VEC X_T, \VEC U_T)
      \mid Z_T = z_T, \VEC \Theta_T = \VEC \theta_T, \tilde{\VEC D}_T = \tilde{\VEC d}_T ] \\
   \intertext{and for $t=T-1,T-2,\dots,1$,}
   V_t(z_t, \VEC \theta_t) &= \min_{\tilde{\VEC d}_t}
   \EXP [ c_t(Z_t, \VEC X_t, \VEC U_t) + V_{t+1}(\Pi_{t+1})
      \mid Z_t = z_t, \VEC \Theta_t = \VEC \theta_t, \tilde{\VEC D}_t = \tilde{\VEC d}_t ]
  \end{align}
  Let $\tilde \Psi_t(z_t, \VEC \theta_t)$ denote the $\arg \min$ of the right hand side of
  $V_t(z_t, \VEC \theta_t)$, and $\tilde \Psi^i_t$ denote the $i$-th component of $\tilde \Psi_t$. 
  Then, a control strategy
  \[ \tilde g^i_t(x^i_t, z_t, \VEC \theta_t) \in \tilde \Psi^i_t(z_t, \VEC \theta_t)(x^i_t) \]
  is optimal for Problem~\ref{prob:finite} with the full observation model.
\end{theorem}

\section {Main result for the partial observation model} \label{sec:partial}

In this section, we derive structure of optimal control laws and a dynamic
programming decomposition for the partial observation model. As in the full
observation model, we cannot directly use the results
of~\cite{NayyarMahajanTeneketzis:2012} because the local observations
$Y^i_{1:t}$ at each control station are increasing with time. To circumvent this
difficulty we follow a three step approach, similar to the one taken for the
full observation model, and proceed as follows:
\begin{enumerate}
  \item Use a person-by-person approach to show that $\Xi^i_t(x) \DEFINED
    \PR(X^i_t = x \mid Y^i_{1:t}, Z_{1:t}, \VEC U_{1:t-1})$ is a sufficient
    statistic for the history of local observations at control station $i$ at
    time $t$. Thus, for any control strategy of control station $i$ that uses
    $(Y^i_{1:t}, Z_{1:t}, \VEC U_{1:t-1})$, we can choose a strategy that uses
    only $(\Xi^i_t, Z_{1:t}, \VEC U_{1:t-1})$ without loss of optimality.
  \item Steps 2 and 3 are similar to those of the full observation model with
    $X^i_t$ replaced by $\Xi^i_t$.
\end{enumerate}

Now, we describe each of these steps in detail.

\subsection*{Step 1: Sufficient statistic for local observations}

In this step, we find a sufficient statistic for the local observations
$Y^i_{1:t}$ at control station~$i$. For that matter, we define the following:

\begin{definition}
  \label{def:xi}
  Given any control strategy $\VEC g$ of the form~\eqref{eq:M2}, let $\Xi^i_t$,
  $i=1,\dots,n$, $t=1,\dots,T$ denote the posterior probability of the local
  state $X^i_t$ of substation~$i$ given all the information $(Y^i_{1:t},
  Z_{1:t}, \VEC U_{1:t-1})$ at control station~$i$, i.e., for any $x^i \in
  \ALPHABET X^i$, the component $x^i$ of $\Xi^i_t$ is given by
  \[ \Xi^i_t(x^i) \DEFINED \PR^{\VEC g}(X^i_t = x^i \mid Y^i_{1:t}, Z_{1:t}, \VEC
    U_{1:t-1}) \stackrel{(a)}= 
    \PR^{\VEC g}(X^i_t = x^i \mid Y^i_{1:t}, Z_{1:t-1}, \VEC U_{1:t-1})
    \]
    where $(a)$ follows from the independence of $\{W^0_t$, $t=1,\dots,T\}$ from
    $\{W^i_t$, $t=1,\dots,T\}$.
\end{definition}

The update of $\Xi^i_t$ follows a non-linear filtering equation as shown below.

\begin{lemma}\label{lemma:xi-update}
  For every $i$, $i=1,\dots,T$, there exist a deterministic functions $\tilde
  F^i_t$ such that
  \begin{equation} \label{eq:xi-update}
    \Xi^i_{t+1} = \tilde F^i_t(\Xi^i_t, Y^i_{t+1}, Z_{t}, \VEC U_t).
  \end{equation}
\end{lemma}

The proof follows from the law of total probability and Bayes rule and is
similar to the proof of Appendix~\ref{app:theta-update}.

The main result of this section is the following:
\begin{proposition}\label{prop:partial-structure}
  In the partial observation model, restricting attention to control laws of the
  form
  \begin{equation} \label{eq:partial-controller-xi}
    U^i_t = \bar g^i_t(\Xi^i_t, Z_{1:t}, \VEC U_{1:t-1})
  \end{equation}
  at all control stations~$i$, $i=1,\dots,n$, is without loss of optimality.
\end{proposition}

The intuition behind this result is as follows. Arbitrarily fix the control
strategies $\VEC g_{-i}$ for all control stations other than~$i$. In the full
observation model, $(X^i_t, Z_{1:t}, \VEC U_{1:t-1})$ is a state sufficient for
performance evaluation at control station~$i$ (Lemma~\ref{lemma:structure}). In
the partial observation model, component $X^i_t$ of this state is not observed.
So, the posterior distribution $\Xi^i_t$ on $X^i_t$ given all the data available
at control station~$i$ should be a sufficient statistic for
$X^i_t$~\cite{Striebel:1965}.

To show that the above intuition is true, we need to establish two conditional
independence properties. 

\begin{proposition} \label{prop:partial-independent}
  Proposition~\ref{prop:independent} is also true for 
  the partial observation model for an arbitrary but fixed choice of
  control strategy $\VEC g$ of the form~\eqref{eq:M2}. 
\end{proposition}

\begin{proposition}\label{prop:xi-independent}
  In the partial observation model, the posterior probability $\Xi^i_t$ of the
  local states of all subsystems are conditionally independent given the history
  of shared state and control actions.
  Specifically, for any Borel subsets $E^i_t$ of $\PSP {X^i}$, $\VEC E_t =
  (E^1_t, \dots, E^n_t)$, 
  $u^i_t \in \ALPHABET U^i$, $z_t \in \ALPHABET Z$, $i=1,\dots,n$ and
  $t=1,\dots,T$, we have
  \begin{equation}
    \PR(\VEC \Xi_{1:t} \in \VEC E_{1:t} \,|\, Z_{1:t} = z_{1:t}, \VEC U_{1:t} = \VEC u_{1:t}) 
    = \prod_{i=1}^n 
    \PR(\Xi^i_{1:t} \in E^i_{1:t} \,|\, Z_{1:t} = z_{1:t}, \VEC U_{1:t} = \VEC u_{1:t})
  \end{equation}
\end{proposition}

These results are proved in Appendices~\ref{app:partial-independent}
and~\ref{app:xi-independent}.

An immediate consequence of Proposition~\ref{prop:partial-independent} and
Lemma~\ref{lemma:xi-update} is the following (see Appendix~\ref{app:partial-CMP}
for proof).

\begin{lemma} \label{lemma:partial-CMP}
  Lemma~\ref{lemma:CMP} is also true for the partial observation model with
  $R^i_t$ defined as $(\Xi^i_t, Z_{1:t}, \VEC U_{1:t-1})$. 
\end{lemma}

\begin{proof}[Proof of Proposition~\ref{prop:partial-structure}]
  The result of Proposition~\ref{prop:partial-structure} follows from cyclically
  repeating an argument similar to the argument after Lemma~\ref{lemma:CMP}.
\end{proof}

\subsection* {Steps 2 and 3: Sufficient statistic for common data and its
simplification}

Compare Proposition~\ref{prop:structure} of the full observation model with
Proposition~\ref{prop:partial-structure} of the partial observation model. The
posterior probability $\Xi^i_t$ in the latter model plays the role of local
state $X^i_t$ in the former model. This suggests that we may follow Steps 2 and
3 of the full observation model in the partial observation model by replacing
$X^i_t$ by $\Xi^i_t$. Following this suggestion, define:

\begin{definition} \label{def:partial-pi}
  Let
  $\mathring \Pi_t$ denote the posterior probability on $(Z_t, \VEC \Xi_t)$
  given the common information $C_t$, i.e., for any $z \in \ALPHABET Z$ and any
  Borel subsets $E^i$ of $\PSP {X^i}$ and $\VEC E = (E^1, \dots, E^n)$, 
  \begin{equation}
    \mathring \Pi_t(z, \VEC E) = \PR(Z_t = z, \VEC \Xi_t \in \VEC E \mid C_t)
  \end{equation}
\end{definition}
\begin{definition} \label{def:partial-theta}
  Let $\mathring \Theta^i_t$, $t=1,\dots,T$, denote the
  posterior probability of $\Xi^i_t$ given the common information $(Z_{1:t}, \VEC
  U_{1:t-1})$, i.e., for any Borel subset $E^i$ of $\PSP {X^i}$, 
  \[ \mathring \Theta^i_t(E^i) \DEFINED \PR( \Xi^i_t \in E^i \mid Z_{1:t}, \VEC U_{1:t-1}). \]
\end{definition}

Now, by following the exact same argument as in Steps 2 and 3 for the full
observation model, we get that Propositions~\ref{prop:coordinator-structure}
and~\ref{prop:coordinator-DP} and Theorems~\ref{thm:structure} and~\ref{thm:DP}
are also true for the partial observation model if we replace $\Pi_t$ and
$\Theta^i_t$ by $\mathring \Pi_t$ and $\mathring \Theta^i_t$, respectively.

\section {Extension to infinite horizon} \label{sec:infinite}

In this section, we extend the result of structural result of
Theorem~\ref{thm:structure} and the dynamic programming decomposition of
Theorem~\ref{thm:DP} to a time-homogeneous system that runs for an infinite
horizon under the discounted cost optimality criterion.

In the model of Section~\ref{sec:model}, assume that the plant function $f^i_t$,
$i=0,\dots,n$, and the cost function $c_t$ are time-invariant and are denoted by
$f^i$ and $c$, respectively. Furthermore, in the partial observation model
assume that the observation function $\ell^i_t$, $i=1,\dots,n$ are
time-invariant and are denoted by $\ell^i$. Such a system is called a
\emph{time-homogeneous system}. 

Assume that the system runs indefinitely. Define the performance of a control
strategy $\VEC g \DEFINED (\VEC g_1, \VEC g_2, \dots)$ as
\begin{equation}
  J_\beta(\VEC g) \DEFINED \lim_{T \to \infty} \EXP\Big[
  \sum_{t=1}^T \beta^{t-1} c(Z_t, \VEC X_t, \VEC U_t) \Big],
  \label{eq:cost-discount}
\end{equation}
where $\beta \in (0,1)$ is called the discount factor.

We are interested in the following optimization problem.
\begin{problem} \label{prob:discounted}
  Given a discount factor $\beta$, the distributions $P_Z$, $P_{X^i|Z}$,
  $P_{W^i}$, $P_{\tilde W^i}$ of the initial shared state, initial local state,
  plant disturbance of subsystem~$i$, and observation noise of subsystem~$i$
  (for the partial observation model), $i=1,\dots,n$, and the cost functions $c$,
  find a control strategy $\VEC g$ that minimizes the expected
  discounted cost given by~\eqref{eq:cost-discount}.
\end{problem}

Since the sufficient statistic in Theorem~\ref{thm:DP} takes value in a
time-invariant space, the results of the finite horizon system extend to
infinite horizon in the usual manner. Proposition~\ref{prop:independent} remains
valid for an infinite horizon system as well. Consequently, so do the
structural results of Proposition~\ref{prop:structure}. Therefore, we can use
the approach of~\cite{NayyarMahajanTeneketzis:2012} to obtain the infinite
horizon version of the dynamic program of Proposition~\ref{prop:coordinator-DP}.
Using Lemma~\ref{lemma:equivalent}, the dynamic program simplifies as follows:

\begin{theorem}\label{thm:infinite}
  There exists an optimal control strategy that is time homogeneous. An
  optimal choice of the partially evaluated control strategy $\tilde d$ of
  $\tilde g$ is given by solution of the following fixed point
  equation\footnote{Due to the discounting of future costs, \eqref{eq:DP-inf}
  has a fixed point that is unique.} for the full observation model:
  \begin{equation}
    V(z, \VEC \theta) = \min_{\tilde{\VEC d}} \EXP \Big[
      c(Z_t, X_t, U_t) + \beta V(Z_{t+1}, \VEC \Theta_{t+1})
      \Bigm| Z_t = z, \VEC \Theta_t = \VEC \theta, \tilde{\VEC D}_t =
      \tilde{\VEC d} \Big]
      \label{eq:DP-inf}
  \end{equation}
  and, by replacing $\Theta^i_t$ by $\mathring \Theta^i_t$ in the above equation
  for the partial observation model. (The above equation is time homogeneous; we
  are using time $t$ for ease of notation.)
\end{theorem}

\section {An example: Multiaccess broadcast} \label{sec:MAB}

In this section, we reconsider the multiaccess broadcast system described in
Section~\ref{sec:applications} and show how the results of this paper provide
new insights for that system.

\subsection {The model} \label{sec:MAB-model}

Recall that a two-user multiaccess system consists is a special case of the full
observation model with $\ALPHABET X^i = \ALPHABET U^i = \ALPHABET W^i_t =
\{0,1\}$, $i=1,2$, and
$\ALPHABET Z = \emptyset$. The state dynamics of user~1 are given by:
\( X^1_{t+1} = \max( X^1_t + U^1_t \cdot (1- U^2_t) + W^1_t, 1). \)
The dynamics of user~2 are symmetric dual of the above.  Each user chooses a
transmission decision as
\( U^i_t = g^i_t(X^i_{1:t}, \VEC U_{1:t-1}) \)
where only actions $U^i_t \le X^i_t$ are feasible.
The per unit reward function $c(\VEC x, \VEC u) = u^1
\oplus u^2$, where $\oplus$ means binary XOR. The objective is to maximize the
total average reward over an infinite horizon given by
\begin{equation}
  \bar J(\VEC g) = \lim_{T \to \infty} \frac 1T \EXP \Big[
    \sum_{t=1}^T U^1_t \oplus U^2_t \Big]
  \label{eq:MAB-reward}
\end{equation}
which corresponds to maximizing the average throughput.

The case of symmetric arrivals ($p^1 = p^2)$ was considered
in~\cite{HluchyjGallager:1981}, who found a lower bound on performance by
finding the best \emph{window protocol} strategies. An upper bound was for the
symmetric case was computed numerically in~\cite{OoiWornell:1996} by
considering a more informative information structure. The analytic lower bounds
of~\cite{HluchyjGallager:1981} match the numerical upper bound
of~\cite{OoiWornell:1996}; hence, the strategy proposed
in~\cite{HluchyjGallager:1981} is optimal. A dynamic programming decomposition
for the general model was presented in~\cite{MahajanNayyarTeneketzis:2008}.


The multiaccess broadcast system corresponds to the full observation model.
Therefore, the results of this
paper\footnote{\label{fnt:average} Although in Section~\ref{sec:infinite}, we
only considered the infinite horizon discounted cost criterion, the same
argument also works for the average reward per unit time.} provide a structure
of optimal transmission policies and a dynamic programming decomposition. For
the symmetric arrival case ($p^1 = p^2$), we solve the corresponding dynamic
program in closed form, and give an analytic derivation of the optimal strategy.

\subsection {Structure of optimal transmission policies and dynamic programming
decomposition}
\label{sec:MAB-control-sharing}

Since $Z_t = \emptyset$, the information state $\VEC \Theta_t = (\Theta^1_t,
\dots, \Theta^n_t)$ of Definition~\ref{def:theta} simplifies to
\( \Theta^i_t(x) = \PR^{\tilde {\VEC g}}(X^i_t = x \mid \VEC U_{1:t-1}). \)
Theorem~\ref{thm:structure} implies that there is no loss of optimality
in restricting attention to control strategies of the form
\(U^i_t = g^i_t(X^i_t, \VEC \Theta_t)\)
and Theorem~\ref{thm:DP} gives the corresponding dynamic program to find the
optimal transmission strategies. 

To succinctly describe the dynamic program, we simplify the notation as follows:
\begin{enumerate}
  \item The functional map $\tilde d^i_t$ from $\ALPHABET X^i$ to $\ALPHABET
    U^i$ is completely specified by $\tilde d^i_t(1)$ because $\tilde d^i_t(0)$
    must be zero as $u^i_t = \tilde d^i_t(x^i_t) \le x^i_t$ and $\ALPHABET X^i =
    \ALPHABET U^i = \{0,1\}$. We denote $\tilde d^i_t(1)$ by $S^i_t \in \{0,1\}$.
    Then, \( u^i_t = x^i_t \cdot s^i_t. \)
  \item Since $\Theta^i_t$ is a probability distribution of a binary valued
    random variable, it is completely specified by its component
    $\Theta^i_t(1)$, which we denote by $Q^i_t$.
\end{enumerate}

To present the update of $\VEC Q_t$, we define the following operators.
\begin{definition} \label{def:A}
  Let $A_i$, $i=1,2$, be an operator from $[0,1]$ to $[0,1]$ defined for any $q \in [0,1]$
  as
  \( A_i q = 1 - (1-p^i)(1-q) \)
  where $p^i$ is the arrival rate at the queue $i$. Then, \( A_i^n q = 1 -
  (1-p^i)^n (1-q)\), and for any $q \in (0,1)$, $A^{n}_i q < A^{n+1}_i q$. 
\end{definition}

Lemma~\ref{lemma:theta} shows that the 
information state $\VEC q_t = (q^1_t, q^2_t)$ updates according to a non-linear
filter \( F(\VEC q_t, \VEC u_t,
\VEC s_t) \)
where
\begin{equation} \label{eq:simplified-F}
 F( (q^1,q^2), \VEC u, \VEC s) = 
  \begin{cases}
    (A_1q^1, A_2q^2), & \text{if $\VEC s = (0,0)$} \\
    (   p^1, A_2q^2), & \text{if $\VEC s = (1,0)$} \\
    (A_1q^1,  p^2), & \text{if $\VEC s = (0,1)$} \\
    (   1,    1), & \text{if $\VEC s = (1,1)$ and $\VEC u = (1,1)$} \\
    ( p^1,  p^2), & \text{if $\VEC s = (1,1)$ and $\VEC u \neq (1,1)$}.
  \end{cases}
\end{equation}
Substituting this update function in the infinite horizon average reward
per unit time version of the dynamic program of
Theorem~\ref{thm:DP}, we get
\begin{proposition} \label{prop:MAB-DP}
  For the two-user multiaccess broadcast system, there is no loss in optimality
  in restricting attention to time-homogeneous transmission strategies of the form
  \[ U^i_t = \tilde g^i_t(X^i_t, \VEC Q_t) = S^i_t(\VEC Q_t) \cdot X^i_t. \]
  An optimal strategy of such form is given by the solution of the following fixed point
  equation:
  \begin{equation}
    v(q^1, q^2) + J^* = \max \{ v_{10}(q^1, q^2), v_{01}(q^1, q^2), v_{11}(q^1,
  q^2) \}
  \end{equation}
  where $J^*$ denotes the average reward per unit time, $v(q^1, q^2)$ is the
  relative value function at $(q^1, q^2)$ and $v_{ij}(q^1, q^2)$ is the 
  relative value-action function at $(q^1,q^2)$ when $(s^1, s^2)$ is chosen to be $(i,j)$,
  $i, j \in \{0,1\}$, \emph{i.e.}, 
  \begin{gather*}
    v_{10}(q^1,q^2) = q^1 + v(p^1, A_2 q^2), \\ 
    v_{01}(q^1,q^2) = q^2 + v(A_1 q^1, p^2),
    \\
    v_{11}(q^1,q^2) = q^1 + q^2 - 2 q^1 q^2 + q^1 q^2 v(1,1) + (1-q^1q^2)
    v(p^1,p^2).
  \end{gather*}
\end{proposition}

\paragraph*{Some remarks}
  \begin{enumerate}
    \item We ruled out the action $(s^1, s^2) = (0,0)$ because it is dominated
      by the action $(s^1, s^2) = (1,0)$. 

    \item The information state $\VEC q$ takes values in the uncountable set
      $[0,1]^2$. However, the form of the non-linear filter $F$~\eqref{eq:simplified-F}
      implies that the reachable set of $\VEC q$ is countable and is given by
      \[ \ALPHABET R = \{ (1,1), (1,p^2), (p^1,1), (p^1,p^2) \} 
          \cup
          \{ (p^1, A_2^n p^2) : n \in \naturalnumbers \}
          \cup
          \{ (A_1^np_1, p^2) : n \in \naturalnumbers \}
      \]
      Thus, we need to solve the dynamic program of Proposition~\ref{prop:MAB-DP} only for $\VEC
      q \in \ALPHABET R$. 
    \item A similar dynamic programming decomposition for the two-user
      multiaccess broadcast channel was derived
      in~\cite{MahajanNayyarTeneketzis:2008},
      but~\cite{MahajanNayyarTeneketzis:2008} did not completely exploit the
      information structure of the system. In particular,
      \cite{MahajanNayyarTeneketzis:2008} \emph{a priori} restricted attention
      to transmission strategies of the form~\eqref{eq:structure-local} while we
      show that such a restriction is without loss of optimality. Furthermore,
      the dynamic program in~\cite{MahajanNayyarTeneketzis:2008} is similar to
      that of Proposition~\ref{prop:coordinator-DP} while we use a simpler form
      of the dynamic program (Theorem~\ref{thm:DP}). As shown above, the reachable
      set of the information state is countable for this simpler form of the dynamic program
      while such a simplification was not possible for the dynamic
      program in~\cite{MahajanNayyarTeneketzis:2008}.
  \end{enumerate}

\subsection{The symmetric arrival case} Assume that both users have symmetric
arrivals, \emph{i.e.}, $p^1 = p^2$. Then the transformation $A_1$ is the same as
$A_2$, and we denote both by $A$. Since the system is symmetric for both users,
we have that for any $q^1, q^2$, the relative value functions $v(q^1, q^2)$ and
$v(q^2, q^1)$ are the same. Therefore $v_{ij}(q^1, q^2) = v_{ji}(q^2, q^1)$, and
consequently, the optimal coordination policy is also symmetric, \emph{i.e.}, 
$h(q^1, q^2) = h(q^2, q^1)$.
      
Using this symmetry, we find a closed form solution of the dynamic program. 
To describe the solution, we first consider the following polynomial and some of its
properties. Let
\( \varphi_n(x)  = 1 + (1-x)^2 - (3+x)(1-x)^{n+1}. \)
Note that
\begin{enumerate}
  \item $\varphi_n(0) = -1$ and $\varphi_n(1) = 1$. Thus, $\varphi_n$
    has a root $\alpha_n$ that lies in the interval $[0,1]$. 
  \item $\varphi_{n+1}(x) = (1-x) \varphi_n(x) + x(1 + (1-x)^2)$. Thus,
    $\varphi_{n+1}(\alpha_n)$ is positive. Recall that $\varphi_n(0) = -1$.
    Thus, $\alpha_{n+1}$ lies in the interval $[0, \alpha_n]$. Hence the
    sequence $\{\alpha_n\}$ is decreasing. 
\end{enumerate}

Let $\tau$ denote the root of $x = (1-x)^2$. Then, $\tau \approx 0.38196 >
\alpha_1 \approx 0.34727$.

\begin{theorem} \label{thm:MAB}
  For the symmetric arrival case, $p^1 = p^2 = p$, the optimal solution $J^*$ to the
  dynamic program of Proposition~\ref{prop:MAB-DP} is given by
  \begin{equation}
    J^* = \begin{cases}
      (1 - (1-p)^2), & \text{if $p \ge \alpha_1$}, \\
      p(1-(2p^2 - 1))/(1 + p^2 + p^3), & \text{otherwise}.
    \end{cases}
  \end{equation}
  The corresponding optimal strategy $h^*(q^1,q^2)$, $(q^1, q^2) \in \ALPHABET
  R$,  is given by
  \begin{enumerate}
    \item For $p \ge \tau$, 
      \[ h^*(q^1, q^2) = \begin{cases}
          (1,0), & \text{if $q^1 > q^2$}, \\
          (0,1), & \text{if $q^1 < q^2$}, \\
          (1,0) \text{ or } (0,1),& \text{if $q^1 = q^2$}.
        \end{cases} \]
    \item For $p < \tau$, let $n \in \naturalnumbers$ be such that
      $\alpha_{n+1} < p \le \alpha_n$. Then,
      \[ h^*(q^1, q^2) = \begin{cases}
          (1,1), & \text{if $q^1 \le A^np$ and $q^2 \le A^np$}, \\
          (1,0), & \text{if $q^1 > \max(A^n p, q^2)$}, \\
          (0,1), & \text{if $q^2 > \max(A^n p, q^1)$}, \\
          (1,0) \text{ or } (0,1),& \text{if $q^1 = q^2 = 1$}.
        \end{cases} \]
  \end{enumerate}
\end{theorem}
The proof is presented in Appendix~\ref{app:MAB}.

Although the optimal policy looks complicated with different behavior depending
on the value of $p$, it has only two modes of operation. When $p \ge
\tau$, the set of states $\{(p,Ap), (Ap,p)\}$ is absorbing and forms a
recurrence class in $\ALPHABET R$. Within this recurrence class, the optimal
policy is a round-robin policy. When $p < \tau$, the set of states $\{(1,1), (p,
Ap), (Ap, p), (p,p)\}$ is absorbing and forms a recurrence class in $\ALPHABET
R$. Within this recurrence class, the optimal policy is identical for all $p <
\tau$. The system starts with $(q^1, q^2) = (p,p)$ and chooses $(s^1, s^2) =
(1,1)$, which means that each user transmits if it has a packet. If no collision
occurs, then the next state remains $(p,p)$. If a collision occurs, $(q^1, q^2)
= (1,1)$ and both users know that both of them have a packet. So, they simply
empty their buffer one by one, say first $(s^1,s^2) = (1,0)$, and then $(s^1,
s^2) = (0,1)$, and go back to ``transmit if you have a packet'' action:
$(s^1,s^2) = (1,1)$. This policy is identical to the optimal window protocol
proposed in~\cite{HluchyjGallager:1981}. Unlike~\cite{HluchyjGallager:1981}, who
showed that this strategy is the best transmission strategy when restricted to
window protocols, we have shown that this strategy is the best strategy over
the class of all transmission protocols. 

\section {Discussion and Conclusion} \label{sec:conclusion}

Systems with control sharing information structure arise in a variety of
communication applications. In this paper, we presented a three step approach to
identify sufficient statistic and dynamic programming decomposition for coupled
subsystems with control sharing.

The general decentralized control system with control sharing does not admit a
tractable dynamic programming decomposition. Our solution approach works because
the subsystems are coupled \emph{only} through control actions $\VEC U_t$ and
shared state $Z_t$, \emph{but not through local states $\VEC X_t$}. In particular,
if the system dynamics were of the form
\begin{equation} \label{eq:general-dynamics}
  X^i_{t+1} = f^i_t(Z_t, \VEC X_t, \VEC U_t, W^i_t) 
\end{equation}
instead of~\eqref{eq:dynamics}, then Propositions~\ref{prop:independent}
and~\ref{prop:partial-independent} will fail, and consequently, Step~1 of our
approach would not simplify the control strategies. 

In addition, the final sufficient statistics $\VEC \Theta_t$ and $\mathring
{\VEC \Theta}_t$ derived in Step~3 are simpler than the general sufficient
statistics $\Pi_t$ and $\mathring \Pi_t$, which are based
on~\cite{NayyarMahajanTeneketzis:2011}, derived in Step~2. In particular, $\Pi_t
\in \PSP*{\ALPHABET X^1 \times \cdots \times \ALPHABET X^n}$, so its size
increases exponentially with the number of subsystems, while $\mathring \Pi_t
\in \PSP {X^1} \times \cdots \times \PSP {X^n}$, so its size increases linearly
with the number of subsystems. This additional simplification is also a
consequence of the specific form of system dynamics and would fail if the system
dynamics were of the form~\eqref{eq:general-dynamics}.

In itself, it is not surprising that a simpler dynamical model makes the system
easier to design. However, it is important to understand why \emph{this
particular} simplification dynamical model works; such an understanding will
allow for similar simplifications for general non-classical information
structures as well.

The system dynamics given by~\eqref{eq:dynamics} do not remove the incentive to
signal. In particular, control station~$i$ at time~${t+1}$ does not know all
observations of control station~$j$ at time~$t$. Hence, control station~$j$ has
an incentive to signal its local observation to control station~$i$ through its
action~$U^j_t$. Thus, the model is not partially nested~\cite{HoChu:1972} (or
quasi-classical~\cite{MahajanYuksel:2010}). Even after taking the conditional
independence results of Propositions~\ref{prop:independent}
and~\ref{prop:partial-independent} into account, the signaling incentive is
still present due to the cost coupling. Knowing the local state $X^j_t$ of
subsystem~$j$ will help control station~$i$ to improve its choice of action
$X^i_t$ in order to minimize the expected cost to go $\EXP[ \sum_{s=t}^T
c_s(Z_s, \VEC X_s, \VEC U_s) ]$. Thus, the model is not stochastically
nested~\cite{Yuksel:2009} (or $P$-quasi-classical~\cite{MahajanYuksel:2010}).

We may think of the system dynamics of the form~\eqref{eq:dynamics} as a
sufficient condition to obtain a time-invariant sufficient statistic for the
local information at each control stations (Propositions~\ref{prop:structure}
and~\ref{prop:partial-structure}). Once such a sufficient statistic is
identified, the model reduces to a partial history sharing information structure
with \emph{local information taking values in a time-invariant space}.
Thereafter, one can use the results of~\cite{NayyarMahajanTeneketzis:2011} 
obtain a sufficient statistic of the common information at all control stations.

Finding such sufficient conditions (to extend the applicability of a specific
solution technique to more general models) is a recurring theme in decentralized
control. A similar approach has been used in~\cite{Mahajan:phd} to generalize
the solution approach of~\cite{Witsenhausen:1973} to two controller
teams where at least one controller has finite memory; in~\cite{Yuksel:2009} to
generalize the solution approach of~\cite{HoChu:1972} to stochastically nested
information structures; in~\cite{WuLall:2010} to generalize the solution
approach of~\cite{MahajanNayyarTeneketzis:2008} to broadcast information
structures; and in~\cite{MahajanYuksel:2010} to generalize the solution of
classical and quasiclassical information structures to $P$-classical and
$P$-quasiclassical information structures. The model and results of this paper
present such a sufficient condition to extend the results
of~\cite{NayyarMahajanTeneketzis:2011} to control sharing information structure.

\section* {Acknowledgment}

This work was supported by the Natural Sciences and Engineering Research Council of Canada
through the grant NSERC-RGPIN 402753-11. The author is grateful to Ashutosh Nayyar, Demosthenis
Teneketzis, and Serdar Y\"uksel for helpful discussions.

\bibliographystyle{IEEEtran}
\bibliography{IEEEabrv,../../collection,../../personal}

\appendices

\section {Proof of Proposition~\ref{prop:independent}}
\label{app:independent}

For simplicity of notation, we use $\PR(z_{1:t}, \VEC x_{1:t}, \VEC u_{1:t})$ to denote
$\PR(Z_{1:t} = z_{1:t}, \VEC X_{1:t} = \VEC x_{1:t} , \VEC U_{1:t} = \VEC u_{1:t})$ and a similar
notation for conditional probability. Define:
\begin{itemize}
  \item
  \(\alpha^i_t \DEFINED \PR(u^i_t \,|\, z_{1:t}, x^i_{1:t}, \VEC u_{1:t-1})\),
  \(\beta^i_t  \DEFINED \PR(x^i_t \,|\, z_{t-1}, x^i_{t-1}, \VEC u_{t-1})\),
  \(\gamma^i_t \DEFINED \PR(z_t   \,|\, z_{t-1}, \VEC u_{t-1})\);
  and
  \item 
  \(A^i_t \DEFINED \prod_{s=1}^t \alpha^i_s\),
  \(B^i_t \DEFINED \prod_{s=1}^t \beta^i_s\),
  \(\Gamma_t \DEFINED \prod_{s=1}^t \gamma_s\).
\end{itemize}

From law of total probability it follows that:
\(
  \PR(z_{1:t}, \VEC x_{1:t}, \VEC u_{1:t}) = \bigg(\prod_{i=1}^n A^i_t B^i_t\bigg)
  \Gamma_t.
\)
Summing over all realizations of $\VEC x_{1:t}$ and observing that 
$A^i_t$ and $B^i_t$ depends only on $(z_{1:t}, x^i_{1:t}, \VEC u_{1:t})$, we get
\begin{equation*}
  \PR(z_{1:t}, \VEC u_{1:t}) = \sum_{x^1_{1:t}} \sum_{x^2_{1:t}} \cdots \sum_{x^n_{1:t}}
                       \bigg( \prod_{i=1}^n A^i_t B^i_t \bigg) \Gamma_t 
  = \Bigg( \prod_{i=1}^n \bigg( \sum_{x^i_{1:t}} A^i_t B^i_t \bigg)\Bigg)
  \Gamma_t.
\end{equation*}

Thus, using Bayes rule we get
\begin{equation} \label{eq:independent-1}
  \PR(\VEC x_{1:t} \,|\, z_{1:t}, \VEC u_{1:t}) = \prod_{i=1}^n 
  \dfrac {A^i_t B^i_t}
  {\bigg( \sum_{x^i_{1:t}} A^i_t B^i_t \bigg)}
\end{equation}
Summing both sides over $x^i_{1:t}$, $i \neq j$, we get
\begin{equation} \label{eq:independent-2}
  \PR(x^j_{1:t} \,|\, z_{1:t}, \VEC u_{1:t}) = 
  \dfrac {A^j_t B^j_t}
  {\bigg( \sum_{x^j_{1:t}} A^j_t B^j_t \bigg)}
\end{equation}

The result follows from combining~\eqref{eq:independent-1}
and~\eqref{eq:independent-2}.

\section {Proof of Lemma~\ref{lemma:CMP}}
\label{app:CMP}

For ease of notation, we use $\PR(\tilde r^i_{t+1} \,|\, r^i_{1:t}, u^i_{1:t})$
to denote $\PR( {R^i_{t+1} = \tilde r^i_{t+1}} \,|\, {R^i_{1:t} = r^i_{1:t}},
{U^i_{1:t} = u^i_{1:t}})$ and a similar notation for other probability
statements.  
Consider
\begin{multline}
 \PR(\tilde r^i_{t+1} \,|\, r^i_{1:t}, u^i_{1:t}) 
 = \PR(\tilde x^i_{t+1} \,|\, x^i_t, \tilde z_t, \tilde {\VEC u}_t) 
    \cdot
    \PR(\tilde z_{t+1} \,|\, \tilde z_t, \tilde {\VEC u}_t)
   \cdot
   \IND[ \tilde {\VEC u}_{1:t-1} = \VEC u_{1:t-1} ] 
   \cdot
   \IND[\tilde u^i_t = u^i_t] 
     \\
    \cdot
    \IND[ \tilde z_{1:t} = z_{1:t} ] 
    \cdot
   \PR(\tilde {\VEC u}^{-i}_t \,|\, x^i_{1:t}, z_{1:t}, \VEC u_{1:t-1}, u^i_t) 
    \label{eq:CMP-1}
\end{multline}

Simplify the last term of~\eqref{eq:CMP-1} as follows:
\begin{align}
  \hskip 2em & \hskip -2em  
  \PR(\tilde {\VEC u}^{-i}_t \,|\, x^i_{1:t}, z_{1:t}, \VEC u_{1:t-1}, u^i_t)
  \stackrel{(a)}{=}
  \PR(\tilde {\VEC u}^{-i}_t \,|\, x^i_{1:t}, z_{1:t}, \VEC u_{1:t-1}) \notag \\
  &= \sum_{\VEC x^{-i}_{1:t}} \PR( \tilde {\VEC u}^{-i}_t \,|\, \VEC x^{-i}_{1:t},
  z_{1:t}, \VEC u_{1:t-1}) 
  \cdot 
  \PR( \VEC x^{-i}_{1:t} \,|\, x^i_{1:t}, z_{1:t}, \VEC u_{1:t-1})
  \notag \\
  &\stackrel{(b)}=
  \sum_{\VEC x^{-i}_{1:t}} \PR( \tilde {\VEC u}^{-i}_t \,|\, \VEC x^{-i}_{1:t},
  z_{1:t}, \VEC u_{1:t-1}) 
  \cdot
  \PR( \VEC x^{-i}_{1:t} \,|\, z_{1:t}, \VEC u_{1:t-1})
  = 
  \PR(\tilde {\VEC u}^{-i}_t \,|\, z_{1:t}, \VEC u_{1:t-1})
  \label{eq:CMP-2}
\end{align}
where $(a)$ is true because $u^i_t$ is determined by $x^i_{1:t}$ $z_{1:t}$ and $\VEC
u_{1:t-1}$ and $(b)$ follows from Proposition~\ref{prop:independent}.
Substituting~\eqref{eq:CMP-2} in~\eqref{eq:CMP-1}, we get
\begin{align}
 \hskip 2em & \hskip -2em  
 \PR(\tilde r^i_{t+1} \,|\, r^i_{1:t}, u^i_{1:t}) 
 = \PR(\tilde x^i_{t+1} \,|\, x^i_t, \tilde z_t, \tilde {\VEC u}_t) 
    \cdot
    \PR(\tilde z_{t+1} \,|\, \tilde z_t, \tilde {\VEC u}_t)
    \cdot
    \IND[ \tilde {\VEC u}_{1:t-1} = \VEC u_{1:t-1} ] 
    \notag \\
 & \hskip 6em 
    \cdot
    \IND[ \tilde z_{1:t} = z_{1:t} ] 
 \cdot \IND[\tilde u^i_t = u^i_t] \cdot
   \PR(\tilde {\VEC u}^{-i}_t \,|\, z_{1:t}, \VEC u_{1:t-1})\notag \\
 &= \PR(\tilde x^i_{t+1}, \tilde z_{1:t+1}, \tilde {\VEC u}_{1:t} \,|\, x^i_{t}, u^i_{t}, 
 z_{1:t}, \VEC u_{1:t-1}) 
 = \PR(\tilde r^i_{t+1} \,|\, r^i_t, u^i_t)
\end{align}
This completes the proof of part~1) of the Lemma.

To prove part~2), it is sufficient to show that
\(
  \PR(\tilde z_t, \tilde {\VEC x}_{t}, \tilde {\VEC u}_{t} \,|\,
      r^i_{1:t}, u^i_{1:t}) =
  \PR(\tilde z_t, \tilde {\VEC x}_{t}, \tilde {\VEC u}_{t} \,|\,
      r^i_{t}, u^i_{t})
\). Consider
\begin{align}
  \hskip 2em & \hskip -2em  
  \PR(\tilde z_t, \tilde {\VEC x}_{t}, \tilde {\VEC u}_{t} \,|\,
      r^i_{1:t}, u^i_{1:t}) 
   =
  \IND[(\tilde z_t, \tilde x^i_{t} , \tilde u^i_t) = ( z_t, x^i_{t} ,  u^i_t)] 
  \cdot
  \PR( \tilde {\VEC x}^{-i}_{t}, \tilde {\VEC u}^{-i}_t \,|\,
       x^i_{1:t}, u^i_t, z_t, \VEC u_{1:t-1}) \notag \\
  &\stackrel{(c)}=
  \IND[(\tilde z_t, \tilde x^i_{t} , \tilde u^i_t) = ( z_t, x^i_{t} ,  u^i_t)] 
  \cdot
  \PR( \tilde {\VEC x}^{-i}_{t}, \tilde {\VEC u}^{-i}_t \,|\,
      z_{1:t}, \VEC u_{1:t-1}) \notag \\
  &= 
  \PR(\tilde {\VEC x}_{t}, \tilde {\VEC u}_{t} \,|\,
      r^i_{t}, u^i_{t}) 
  \label{eq:CMP-3}
\end{align}
where $(c)$ follows from an argument similar
to~\eqref{eq:CMP-2}.\footnote{Recall that $\VEC x^{-i}_t$ denotes the vector
$(x^1_t, \dots, x^{i-1}_t, x^{i+1}_t, \dots, x^n_t)$.} This completes
the proof of part~2) of the Lemma.

\section {Proof of Lemma~\ref{lemma:theta}} \label{app:theta-update}

Consider the system for a particular realization $(z_{1:T}, \VEC x_{1:T}, \VEC
u_{1:T}, \VEC d_{1:T})$ of $(Z_{1:T}, \VEC X_{1:T}, \VEC U_{1:T}, \VEC
D_{1:T})$. For ease of notation, we use $\PR(x^i_{t+1} \mid z_{1:t+1}, \VEC u_{1:t}, \VEC
d_{1:t})$ to denote $\PR(X^i_{t+1} = x^i_{t+1} \mid Z_{1:t+1} = z_{1:t+1}, \VEC
U_{1:t} = \VEC u_{1:t}, \VEC D_{1:t} = \VEC d_{1:t})$. Define
\begin{gather*}
  A(x^i_{t+1}, \VEC x_t, z_{1:t+1}, \VEC u_{1:t}, \VEC d_{1:t}) \DEFINED
  \PR(x^i_{t+1}, \VEC x_t, z_{t+1}, \VEC u_t \mid z_{1:t}, \VEC u_{1:t-1}, \VEC
  d_{1:t}); 
  \\
  B(x^i_{t+1}, \VEC x_t, z_{t+1}, z_t, \VEC d_t, \VEC \theta_t) \DEFINED
  \PR(x^i_{t+1} \mid x^i_t, z_t, \VEC u_t)
  \cdot
  \PR(z_{t+1} \mid \VEC x_t, z_t, \VEC u_t) 
  \cdot 
  \prod_{i=1}^n \theta^i_t(x^i_t).
\end{gather*}

The system dynamics and Proposition~\ref{prop:independent} implies that
\begin{equation} \label{eq:theta-1}
  A(x^i_{t+1}, \VEC x_t, z_{1:t+1}, \VEC u_{1:t}, \VEC d_{1:t}) 
  = B(x^i_{t+1}, \VEC x_t, z_{t+1}, z_t, \VEC d_t, \VEC \theta_t) \IND[u_t = d_t(\VEC x_t)]
\end{equation}
  
Consider component-$i$ of the realization $\VEC \theta_{t+1}$ of $\VEC \Theta_{t+1}$.
\begin{align}
  \theta^i_{t+1}(x^i_{t+1}) &= 
  \PR(x^i_{t+1} \mid z_{1:t+1}, \VEC u_{1:t}, \VEC d_{1:t}) 
  = \smashoperator[l]{\sum_{\{ \VEC x_t : d_t(\VEC x_t) = \VEC u_t\}} }
      \frac { A(x^i_{t+1}, \VEC x_t, z_{1:t+1}, \VEC u_{1:t}, \VEC d_{1:t}) }
            { \sum_{\tilde x^i_{t+1}} A(\tilde x^i_{t+1}, \VEC x_t, z_{1:t+1}, 
                    \VEC u_{1:t}, \VEC d_{1:t}) }
  \notag \\
 &\stackrel{(a)}=
  \smashoperator[l]{\sum_{\{ \VEC x_t : d_t(\VEC x_t) = \VEC u_t\}} }
  \frac{ B(x^i_{t+1}, \VEC x_t, z_{t+1}, z_t, \VEC d_t, \VEC \theta_t)}
       {\sum_{\tilde x^i_{t+1}} B(\tilde x^i_{t+1}, \VEC x_t, z_{t+1}, z_t, \VEC d_t, \VEC
       \theta_t)}
  \BYDEFINITION F^i_t(\VEC \theta_t, z_{t+1}, \VEC u_t, \VEC d_t)(x^i_{t+1})
  \label{eq:theta-2}
\end{align}
where $(a)$ follows from~\eqref{eq:theta-1}. Combining~\eqref{eq:theta-2} for all
$i$, $i=1,\dots,n$, proves the Lemma.

\section {Proof of Proposition~\ref{prop:partial-independent}}
\label{app:partial-independent}

The proof is similar to proof of Proposition~\ref{prop:independent}. As before,
for ease of notation, we use $\PR(z_{1:t}, \VEC x_{1:t}, \VEC y_{1:t}, \VEC
u_{1:t})$ to denote $\PR(Z_{1:t} = z_{1:t}, \VEC X_{1:t} = \VEC x_{1:t}, \VEC
Y_{1:t} = \VEC y_{1:t}, \VEC U_{1:t} = \VEC u_{1:t})$. Define
\begin{itemize}
  \item \(\alpha^i_t = \PR(u^i_t \mid z_{1:t}, y^i_{1:t}, \VEC u_{1:t-1})\), 
    \(\beta^i_t  = \PR(x^i_t \mid z_{t-1}, x^i_{t-1}, \VEC u_{1:t-1})\),
    \(\gamma_t   = \PR(z_t   \mid z_{t-1}, \VEC u_t)\), 
    \(\delta^i_t = \PR(y^i_t \mid x^i_t)\);
    and
  \item
    \(A^i_t \DEFINED \prod_{s=1}^t \alpha^i_s\),
    \(B^i_t \DEFINED \prod_{s=1}^t \beta^i_s\),
    \(\Gamma_t \DEFINED \prod_{s=1}^t \gamma_s\),
    \(\Delta^i_t \DEFINED \prod_{s=1}^t \delta^i_s\).
\end{itemize}

From the law of total probability, it follows that
\(
  \PR(z_{1:t}, \VEC x_{1:t}, \VEC y_{1:t}, \VEC u_{1:t})
  = big( \prod_{i=1}^n A^i_t B^i_t \Delta^i_t \big) \Gamma^i_t.
\)
Sum over the realizations of $\VEC y_{1:t}$ and observe that $A^i_t$ and
$\Delta^i_t$ depend on $\VEC y_{1:t}$ only through $y^i_{1:t}$. This gives,
\begin{align*}
  \PR(z_{1:t}, \VEC x_{1:t}, \VEC u_{1:t}) &=
  \sum_{y^1_{1:t}} \sum_{y^2_{1:t}} \cdots \sum_{y^n_{1:t}}
  \bigg( \prod_{i=1}^n A^i_t B^i_t \Delta^i_t \bigg) \Gamma_t 
  = \bigg( \prod_{i=1}^n \Big( \sum_{y^i_{1:t}} A^i_t \Delta^i_t \Big) B^i_t
  \bigg)  \Gamma_t
\end{align*}

Now sum over $\VEC x_{1:t}$ and observe that $B^i_t$ and $\Delta^i_t$ depend on
$\VEC x_{1:t}$ only through $x^i_{1:t}$. 
\begin{align*}
  \PR(z_{1:t}, \VEC u_{1:t}) &=
  \sum_{x^1_{1:t}} \sum_{x^2_{1:t}} \cdots \sum_{x^n_{1:t}}
   \bigg( \prod_{i=1}^n \Big( \sum_{y^i_{1:t}} A^i_t \Delta^i_t \Big) B^i_t
  \bigg)  \Gamma_t \\
  &= 
  \bigg( \prod_{i=1}^n \Big( \sum_{x^i_{1:t}} B^i_t \Big( \sum_{y^i_{1:t}} A^i_t \Delta^i_t \Big) 
  \Big)\bigg)  \Gamma_t 
\end{align*}
Thus, by Bayes rule, we get
\begin{equation} \label{eq:partial-1}
  \PR(\VEC x_{1:t} \mid z_{1:t}, \VEC u_{1:t}) = \prod_{s=1}^n
  \frac{ 
    B^i_t \Big( \sum_{y^i_{1:t}} A^i_t \Delta^i_t \Big)
  }{
   \sum_{x^i_{1:t}} B^i_t \Big( \sum_{y^i_{1:t}} A^i_t \Delta^i_t \Big)
  }
\end{equation}
Summing both sides over $x^i_{1:t}$, $i \neq j$, we get
\begin{equation} \label{eq:partial-2}
  \PR(x^j_{1:t} \mid z_{1:t}, \VEC u_{1:t}) = 
  \frac{ 
    B^j_t \Big( \sum_{y^j_{1:t}} A^j_t \Delta^j_t \Big)
  }{
   \sum_{x^j_{1:t}} B^j_t \Big( \sum_{y^j_{1:t}} A^j_t \Delta^j_t \Big)
  }
\end{equation}
The result follows from combining~\eqref{eq:partial-1} and~\eqref{eq:partial-2}.

\section {Proof of Proposition~\ref{prop:xi-independent}} \label{app:xi-independent}

Consider
\[
  \PR(\VEC \Xi_{1:t} \in \VEC E_{1:t} \mid z_{1:t}, \VEC u_{1:t})
  =
  \int_{\VEC E_{1:t}} \D \PR(\VEC \xi_{1:t} \mid z_{1:t}, \VEC u_{1:t})
\]

From Proposition~\ref{prop:partial-independent} and law of total probability, we
get
\begin{align*}
  \D \PR(\VEC \xi_{1:t} \mid z_{1:t}, \VEC u_{1:t}) 
  &=
  \sum_{\VEC x_{1:t}, \VEC y_{1:t}}
  \bigg(
    \prod_{i=1}^n
      \D \PR(\xi^i_t \mid y^i_{1:t}, z_{1:t}, \VEC u_{1:t})
      \cdot
      \PR(y^i_{1:t} \mid x^i_{1:t})
      \cdot
      \PR(x^i_{1:t} \mid z_{1:t}, \VEC u_{1:t})
  \bigg) 
  \notag \\
  &= 
    \prod_{i=1}^n
  \bigg(
    \sum_{\VEC x_{1:t}, \VEC y_{1:t}}
      \D \PR(\xi^i_t \mid y^i_{1:t}, z_{1:t}, \VEC u_{1:t})
      \cdot
      \PR(y^i_{1:t} \mid x^i_{1:t})
      \cdot
      \PR(x^i_{1:t} \mid z_{1:t}, \VEC u_{1:t})
    \bigg) 
\end{align*}
which completes the proof of the Proposition.

\section {Proof of Lemma~\ref{lemma:partial-CMP}} \label{app:partial-CMP}

For ease of notation, we use $\D \PR(\tilde r^{i}_{t+1} | r^i_{1:t},  
u^i_{1:t-1})$ to denote $\D \PR(R^i_{t+1} = \tilde r^i_{t+1} | R^i_{1:t} =
r^i_{1:t}, \allowbreak
U^i_{1:t-1} = u^i_{1:t-1})$ and a similar notation
for other probability measures. Consider
\begin{align}
  \D \PR(\tilde r^i_{t+1}\mid r^i_{1:t}, u^i_{1:t}) 
  &= 
  \smashoperator{\sum_{x^i_{t:t+1}, y^i_{t+1}}}
  \IND[\tilde \xi^i_{t+1} = \tilde F^i_t(\xi^i_t, y^i_{t+1}, \tilde z_t, \tilde
    {\VEC u}_t)]
  \cdot
  \PR(y^i_{t+1} \mid x^i_{t+1})
  \cdot
  \PR(x^i_{t+1} \mid x^i_t, \tilde z_t, \tilde {\VEC u}_t)
  \notag\\
  & \quad
  \cdot
  \PR(\tilde z_{t+1} \mid \tilde z_t, \tilde {\VEC u}_t) 
  \cdot
  \IND[\tilde z_{1:t} = z_{1:t}]
  \cdot 
  \IND[\tilde u^i_t = u^i_t] 
  \cdot
  \IND[\tilde {\VEC u}_{1:t-1} = \VEC u_{1:t-1}]
  \notag \\
  & \quad
  \cdot
  \xi^i_t(x^i_t)
  \cdot
  \PR(\VEC u^{-i}_t \mid 
  \xi^i_{1:t}, z_{1:t}, \VEC u_{1:t-1}, u^i_{t})
  \label{eq:PCMP-1}
\end{align}

Simplify the last term of~\eqref{eq:PCMP-1} as follows:
\begin{align}
  \hskip 2em & \hskip -2em  
  \PR(\VEC u^{-i}_t \mid 
  \xi^i_{1:t}, z_{1:t}, \VEC u_{1:t-1}, u^i_{t})
  \notag\\
  &=
  \smashoperator{\sum_{\VEC x^{-i}_{1:t}, \VEC y^{-i}_{1:t}}}
  \PR(\VEC u^{-i} \mid \VEC y^{-i}_{1:t}, z_{1:t}, \VEC u_{1:t-1})
  \cdot
  \PR(\VEC y^{-i}_{1:t} \mid \VEC x^{-i}_{1:t})
  \cdot
  \PR(\VEC x^{-i}_{1:t}\mid
  z_{1:t}, \VEC u_{1:t-1})
  \notag \\
  &=
  \PR(\VEC u^{-i}_t \mid z_{1:t}, \VEC u_{1:t-1})
  \label{eq:PCMP-2}
\end{align}

Substituting~\eqref{eq:PCMP-2} in~\eqref{eq:PCMP-1} and simplifying, we get
part~1) of the Lemma:
\begin{equation}
  \D \PR(\tilde r^i_{t+1} \mid r^i_{1:t}, u^i_{1:t}) 
  =
  \D \PR(\tilde r^i_{t+1} \mid r^i_{t}, u^i_{t}) 
\end{equation}

The proof of part 2) is similar to~\eqref{eq:CMP-3}.

\section {Proof of Theorem~\ref{thm:MAB}} \label{app:MAB}

We introduce a short hand notation that exploits the symmetry of the problem and
the fact that the reachable set $\ALPHABET R$ is countable. Define

\begin{align*}
  v^* &= v(1,1),      & v^0 &= v(p,p),      & v^n &= v(p, A^np),     \ n \in \naturalnumbers, & v^\infty &= v(p,1) \\
  a^* &= v_{10}(1,1), & a^0 &= v_{10}(p,p), & a^n &= v_{10}(p, A^np),\ n \in \naturalnumbers, & a^\infty &= v_{10}(p,1) \\
  b^* &= v_{01}(1,1), & b^0 &= v_{01}(p,p), & b^n &= v_{01}(p, A^np),\ n \in \naturalnumbers, & b^\infty &= v_{01}(p,1) \\
  c^* &= v_{11}(1,1), & c^0 &= v_{11}(p,p), & c^n &= v_{11}(p, A^np),\ n \in \naturalnumbers, & c^\infty &= v_{11}(p,1) 
\end{align*}
Notice that $v_{10}(A^n p, p) = v_{01}(p, A^np) = b^n$ and $v_{01}(A^np, p) =
v_{10}(p, A^np) = a^n$. 

With the above notation, the dynamic program of Proposition~\ref{prop:MAB-DP} can be written as
\begin{equation} \label{eq:DP*}
  v^n + J^* = \max\{a^n, b^n, c^n \}, \quad n \in \{*,0,1,2,\dots,\infty\}
\end{equation}
where
\begin{subequations}\label{eq:DP**}
\begin{align}
  a^*      &= 1 + v^\infty, & b^*      &= 1 + v^\infty, & c^*      &= v^* \\
  a^0      &= p + v^1     , & b^0      &= p + v^1     , & c^0      &= 2p(1-p) + p^2v^* + (1-p^2)v^0 \\
  a^n      &= p + v^{n+1} , & b^n      &= A^n p + v^1 , & c^n      &=
      \begin{multlined}[t]
        p + A^n p - 2p \cdot A^np + p \cdot A^np \cdot v^* \\
        + (1-p \cdot A^np) v^0
      \end{multlined} \\
  a^\infty &= p + v^\infty, & b^\infty &= 1 + v^1     , & c^\infty &= 1-p + pv^* + (1-p)v^0
\end{align}
\end{subequations}
\begin{lemma} \label{lemma:DP}
  A solution of the fixed point equations of~\eqref{eq:DP**} is given by the
  following:
  \begin{enumerate}
    \item For $p \in (\tau,1]$
      \begin{equation}
        J^* = Ap, \quad v^* = 2 - Ap, \quad v^0 = p, \quad v^n = A^np, \ n \in
        \naturalnumbers, \quad v^\infty = 1.
      \end{equation}
    \item For $p \in (\alpha_1, \tau]$
      \begin{equation}
        J^* = Ap, \quad v^* = 2 - Ap, \quad v^0 = 1 - Ap, \quad v^n = A^np, \ n \in
        \naturalnumbers, \quad v^\infty = 1.
      \end{equation}
    \item For $p \in (\alpha_{m+1}, \alpha_m]$, $m \in \naturalnumbers$,
      define $\zeta(x) = 1 + x^2 + x^3$. Then,  
      \begin{subequations}
        \begin{equation}
          J^* = p \left(1 - \frac{\varphi_0(p)}{\zeta(p)} \right), \quad 
          v^* = 2 - J^*, \quad 
          v^0 = 2 - p - \frac{1+(1-p)^2}{\zeta(p)},
          \quad v^\infty = 1.
        \end{equation}
        and
        \begin{equation}
          v^n = \begin{cases}
            w^n, & \text{if $n \le m$}, \\
            A^n p, & \text{if $n > m$};
          \end{cases}
          \quad n \in
          \naturalnumbers, 
        \end{equation}
        where
        \begin{equation*}
          w^n = (1-p) A^{n-1}p \left(\frac {J^*}p - (1 - p)\right) + v^0 
        \end{equation*}
        Note that $w^1 = J^*$.
      \end{subequations}
  \end{enumerate}
\end{lemma}

The proof follows from elementary algebra. For completeness, we include the
details for each case below.

\subsection* {Case 1: $p \in (\tau, 1]$}

We show that the values of Lemma~\ref{lemma:DP} satisfy the dynamic program
of~\eqref{eq:DP*} and~\eqref{eq:DP**}, by considering the four cases separately.
\begin{enumerate}
  \item $a^* = b^* = 2$, $c^* = 2 - Ap$. Hence, 
    either action $(0,1)$ or $(1,0)$ is optimal at state $(1,1)$ and $v^* + J^*
    = a^* = b^* = 2$.
  \item $a^0 = b^0 = p + Ap$, and $b^0 - c^0 = p^2(p - (1-p)^2)$ which is
    positive for $p > \tau$. Recall that $\tau$ is the root of $x = (1-x)^2$.
    Hence, either action $(0,1)$ or $(1,0)$ is optimal at state $(p,p)$ and $v^0
    + J^* = a^0 = b^0 = p + Ap$.
  \item Consider $n \in \naturalnumbers$. $a^n = p + A^{n+1}p$ and $b^n = A^np + Ap$. Thus, $b^n - a^n =
    p(1-p)\cdot A^{n-1}p \ge 0$. Moreover, $b^n - c^n = p^2[(3-p)A^np - 1] >
    p^2[(3-p)p - 1] $ which is positive for $p \in (\tau, 1]$. Thus, the action
    $(0,1)$ is optimal at state $(p, A^np)$ (and by symmetry, the action $(1,0)$
    is optimal at state $(A^np,p)$) for $n \in \naturalnumbers$ and $v^n + J^* =
    b^n = A^np + Ap$.
  \item $a^\infty = 1 + p$, $b^\infty = 1 + Ap$. Thus, $b^\infty > a^\infty$.
    Moreover, $b^\infty - c^\infty = p \cdot Ap \ge 0$. Thus, the action $(0,1)$ is
    optimal at state $(p,1)$ (and by symmetry, the action $(1,0)$ is optimal at
    state $(1,p)$), and $v^\infty + J^* = b^\infty = 1 + Ap$.
\end{enumerate}

\subsection* {Case 2: $p \in (\alpha_1, \tau]$}

We show that the values of Lemma~\ref{lemma:DP} satisfy the dynamic program
of~\eqref{eq:DP*} and~\eqref{eq:DP**}, by considering the four cases separately.
\begin{enumerate}
  \item $a^* = b^* = 2$, $c^* = 2 - Ap$. Hence, 
    either action $(0,1)$ or $(1,0)$ is optimal at state $(1,1)$ and $v^* + J^*
    = a^* = b^* = 2$.
  \item $a^0 = b^0 = p + Ap$ and $c^0 = 1$. Thus, $c^0 - a^0 = (1-p)^2 - p \ge
    0$ for $p \in (\alpha_1, \tau]$. Recall that $\tau$ is the root of $x = (1-x)^2$.
    Hence, action $(1,1)$ is optimal at state $(p,p)$ and $v^0
    + J^* = c^0 = 1$.
  \item Consider $n \in \naturalnumbers$. $a^n = p + A^{n+1}p$ and $b^n = A^np + Ap$. Thus, $b^n - a^n =
    p(1-p)\cdot A^{n-1}p \ge 0$. Moreover, $c^n = (1-p)A^np + p - Ap + 1$ and,
    thus, $b^n - c^n = 2Ap + p \cdot A^np - p -1 > 2Ap + p \cdot A p - p - 1 =
    \varphi_1(p)$, which is positive for $p \in (\alpha_1, t]$. 
    Thus, the action $(0,1)$ is optimal at state $(p, A^np)$ (and by symmetry,
    the action $(1,0)$ is optimal at state $(A^np,p)$) for $n \in
    \naturalnumbers$ and $v^n + J^* = b^n = A^np + Ap$.
  \item $a^\infty = 1 + p$, $b^\infty = 1 + Ap$. Thus, $b^\infty > a^\infty$.
    Furthermore, $c^\infty = 2 - Ap$ and $b^\infty - c^\infty = 2Ap - 1 = -
    \varphi_0(1-p)$, which
    is positive for $p > 1 - \alpha_0$. Since $\alpha_1 > 1 - \alpha_0$, we have
    that for $p \in (\alpha_1, \tau]$, the
    action $(0,1)$ is optimal at state $(p,1)$ (and by symmetry, the action
    $(1,0)$ is optimal at state $(1,p)$) and $v^\infty + J^\infty = b^\infty =
    1 + Ap$. 
\end{enumerate}

\subsection* {Case 3: $p \in (\alpha_{m+1}, \alpha_{m}]$, $m \in \naturalnumbers$}

We show that the values of Lemma~\ref{lemma:DP} satisfy the dynamic program
of~\eqref{eq:DP*} and~\eqref{eq:DP**}, by considering each case separately.
Recall $\zeta(x) = 1 + x^2 + x^3$. 
\begin{enumerate}
  \item $a^* = b^* = 2$ and $c^* = 2 - J^*$. Hence, 
    either action $(0,1)$ or $(1,0)$ is optimal at state $(1,1)$ and $v^* + J^*
    = a^* = b^* = 2$.
  \item $a^0 = b^0 = p + w^1$ and $c^0 = 2p - p^2 J^* + (1-p^2) v^0$.
    Hence, $a^0 - c^0 = - p^2 \varphi_0(p)/\zeta(p)$ which is positive for $p <
    \alpha_0$. Thus, for $p \in (\alpha_{m+1}, \alpha_m]$, either action $(0,1)$
    or $(1,0)$ is optimal at state $(p,p)$ and $v^0 + J^* = c^0$
  \item Consider $n < m$. Then, $c^n - a^n = -p A^np \varphi_0(p)/\zeta(p)$, which is
    positive for $p \in [0,\alpha_0]$ and thus for $p \in (\alpha_{m+1},
    \alpha_m]$. Moreover, $c^n - b^n = - p^2
    \varphi_n(p)/\zeta(p)$ which is positive for $p \in (0, \alpha_n]$ and hence
    for $p \in (\alpha_{m+1}, \alpha_m]$. (Recall that $\alpha_n$ forms an
    decreasing sequence.). Thus, the action $(1,1)$ is optimal at state $(p,
    A^np)$ for $n < m$ (and by symmetry for state $(A^np, p)$ for $n < m$) and
    $v^n + J^* = c^n$.
  \item Consider $n = m$. Then, $a^m = p + A^{m+1}p$ and $c^m - a^m =
    -(1-p)^m\varphi_1(p) + (1-2p-p^3)$. The first term is positive for $p \in
    [0, \alpha_m]$ and since the second term is larger than $\varphi_1(p)$, the
    second term is also positive in that interval. Moreover $c^m - b^m = - p^2
    \varphi_n(p)/\zeta(p)$ which is positive for $p \in [0,\alpha_m]$. Thus, both
    terms are positive for $p \in (\alpha_{m+1}, \alpha_m]$. Hence, the action
    $(1,1)$ is optimal at the state $(p, A^mp)$ (and by symmetry for the state
    $(A^mp,p)$) and $v^m + J^* = c^m$. 
  \item Consider $n > m$. Then, $a^n = p + A^{n+1}p$ and $b^n = A^n p + J^*$.
    Thus, $b^n - a^n = - p \varphi_0(p)/\zeta(p) - p(1-p)^{n+1} > -p
    \varphi_0(p)/\varphi - p (1-p) = -p^2 \varphi_1(p)/\varphi$ 
    which is positive for $p \in [0, \alpha_1]$ and the second term is always positive. Moreover,
    $b^n - c^n = p^2 \varphi_n(p)/\zeta(p)$ which is positive for $p \in
    (\alpha_n, 1]$ and hence for $p \in (\alpha_{m+1}, \alpha_m] \subset
    (\alpha_n, 1]$. (Recall that $\alpha_n$ forms an decreasing sequence.).
    Thus, the action $(0,1)$ is optimal at state $(p, A^np)$, for $n > m$ (and
    by symmetry the action $(1,0)$ is optimal at state $(A^np, p)$ for $n > m$)
    and $v^n + J^* = b^n$.
  \item $a^\infty = p + J^*$ and $b^\infty = 1 + w^1$. Thus, $b^\infty -
    a^\infty = -p\varphi_0(p)/\zeta(p)$, which is positive for $p <
    \alpha_0$. Moreover, $c^\infty = 1 + p - pJ^* + (1-p)v^0$ and, thus,
    $b^\infty - c^\infty = p^2 (1 + (1-p)^2)/\zeta(p)$ which is always positive.
    Thus, for $p \in (\alpha_{m+1}, \alpha_m]$, the action $(0,1)$ is
    optimal at state $(p,1)$ (and by symmetry the action $(1,0)$ is optimal at
    state $(1,p)$) and ${v^\infty + J^\infty = b^\infty = 2J^*}$.
\end{enumerate}


\end{document}